\documentclass[12pt]{article}

\usepackage{mathtools}
\usepackage[english]{babel} 
\usepackage{amsmath,amsfonts,amsthm,amssymb}
\usepackage{enumerate}
\usepackage{graphicx,psfrag,epsf}
\usepackage{array}
\usepackage{chngcntr}
\usepackage{booktabs} 
\usepackage{natbib}
\usepackage[toc,page]{appendix}
\usepackage{caption}
\usepackage{caption}
\usepackage{subcaption}
\usepackage{color}
\usepackage{tabularx}
\usepackage{footnote}
\usepackage{hyperref}
\usepackage{algorithm}
\usepackage[noend]{algpseudocode}

\makeatletter
\newcommand{\distas}[1]{\mathbin{\overset{#1}{\kern\z@\sim}}}%
\newcommand{\bm}[1]{\mathbf{#1}}
\newsavebox{\mybox}\newsavebox{\mysim}
\newcommand{\distras}[1]{%
  \savebox{\mybox}{\hbox{\kern3pt$\scriptstyle#1$\kern3pt}}%
  \savebox{\mysim}{\hbox{$\sim$}}%
  \mathbin{\overset{#1}{\kern\z@\resizebox{\wd\mybox}{\ht\mysim}{$\sim$}}}%
}
\newcolumntype{C}[1]{>{\centering\let\newline\\\arraybackslash\hspace{0pt}}m{#1}}

\newtheorem{theorem}{Theorem}
\newtheorem{definition}{Definition}
\newtheorem{proposition}{Proposition}
\newtheorem{lemma}{Lemma}
\newtheorem{corollary}{Corollary}
\newcommand{\be}{\begin{equation}}
\newcommand{\ee}{\end{equation}}
\newcommand{\bi}{\begin{itemize}}
\newcommand{\ei}{\end{itemize}}
\newcommand{\ben}{\begin{enumerate}}
\newcommand{\een}{\end{enumerate}}
\newcommand{\stb}{\State $\bullet$ \;}
\DeclareMathOperator*{\argmin}{\arg\!\min}
\def\spacingset#1{\renewcommand{\baselinestretch}%
{#1}\small\normalsize} \spacingset{1}

\makeatother

\let\oldbibliography\thebibliography
\renewcommand{\thebibliography}[1]{\oldbibliography{#1}
\setlength{\itemsep}{0pt}} 

\newcommand{\blind}{0}

\begin{document}

\if0\blind
{
  \title{\bf Minimax and minimax projection designs using clustering}
  \author{Simon Mak \quad and \quad V. Roshan Joseph\\
   H. Milton Stewart School of Industrial and Systems Engineering,\\
  Georgia Institute of Technology, Atlanta, GA 30332}
  \maketitle
} \fi

\if1\blind
{
  \bigskip
  \bigskip
  \bigskip
  \begin{center}
    {\LARGE\bf Title}
\end{center}
  \medskip
} \fi

\bigskip
\begin{abstract}
Minimax designs provide a uniform coverage of a design space $\mathcal{X} \subseteq \mathbb{R}^p$ by minimizing the maximum distance from any point in this space to its nearest design point. Although minimax designs have many useful applications, e.g., for optimal sensor allocation or as space-filling designs for computer experiments, there has been little work in developing algorithms for generating these designs, due to its computational complexity. In this paper, a new hybrid algorithm combining particle swarm optimization and clustering is proposed for generating minimax designs on any convex and bounded design space. The computation time of this algorithm scales linearly in dimension $p$, meaning our method can generate minimax designs efficiently for high-dimensional regions. Simulation studies and a real-world example show that the proposed algorithm provides improved minimax performance over existing methods on a variety of design spaces. Finally, we introduce a new type of experimental design called a minimax projection design, and show that this proposed design provides better minimax performance on projected subspaces of $\mathcal{X}$ compared to existing designs. An efficient implementation of these algorithms can be found in the R package \texttt{minimaxdesign}.
\end{abstract}

\vfill

\spacingset{1.45} 
\section{Introduction}

For a desired design space $\mathcal{X} \subseteq \mathbb{R}^p$, a \textit{minimax distance design} (or simply \textit{minimax design}) is the set of points which \textit{minimizes} the \textit{maximum} distance from any point in $\mathcal{X}$ to its nearest design point. In other words, minimax designs provide a uniform coverage of the design space $\mathcal{X}$ in worst-case scenarios, by ensuring every point in $\mathcal{X}$ is sufficiently well-covered by a design point. The emphasis on mitigating worst-case scenarios allows minimax designs to be applied in a wide range of settings. One such application is in the field of computer experiments, where the goal is to construct a computationally cheap emulator of an expensive simulator using a small number of simulation runs. By conducting these simulations at the points of a minimax design, it can be shown \citep{Jea1990} that the resulting emulator minimizes worst-case prediction error. Minimax designs are also useful for sensor allocation. In particular, by placing sensors according to a minimax design, the minimum information sensed at any point can be maximized. This is particularly important in health and safety monitoring (see, e.g., \citealp{Vea2012}), where failure to detect faults in any part of $\mathcal{X}$ may result in catastrophic human or structural loss. Minimax designs are also useful for resource allocation problems for which an equitable distribution of limited resources is desired \citep{Lus1999}.

Despite its many uses, there has been little algorithmic developments for computing minimax designs \citep{Pat2012}. A major reason for this is that, when $\mathcal{X}$ is a continuous space, the minimax objective (introduced later in Section 2) requires evaluating the supremum over an infinite set, which is costly to approximate. Some existing work include the seminal paper on minimax designs by \cite{Jea1990} and the minimax Latin hypercube designs proposed by  \cite{Van2008}, but both papers only consider two-dimensional designs with restricted design sizes. This greatly limits the applicability of these methods in practice. There has also been some work on minimax designs when $\mathcal{X}$ is approximated by a finite set of points. For example, \cite{Jea1995} studied these designs in the context of two-level factorial experiments, and \cite{Tan2013} proposed a set-covering binary integer program (BIP) for computing minimax designs when points restricted to a finite candidate set of size $N<\infty$. As we show later, BIP can be very time-consuming and provides poor minimax designs for high-dimensional regions. In this paper, we propose a hybrid clustering algorithm which can generate near-optimal minimax designs efficiently, both for large design sizes and in high-dimensions. 

Although most clustering-based designs are not intended for minimax use, there are two reasons for discussing and comparing these designs in our paper. First, an understanding of clustering-based designs allows us to better motivate the proposed minimax clustering algorithm. Second, since the proposed algorithm is similar to the popular Lloyd's algorithm \citep{Llo1957,Llo1982} used in k-means clustering, our simulation studies show that many clustering-based designs indeed possess good minimax properties, and it would be worthwhile to use these designs as a comparison benchmark. The use of clustering in experimental design dates back to \cite{Dal1950} and \cite{Cox1957}, who proposed designs for optimal stratified sampling. K-means clustering using Lloyd's algorithm is also employed for generating a variety of designs, such as \textit{principal points} \citep{Flu1990}, \textit{minimum-MSE quantizers} \citep{LBG1980} and \textit{mse-rep-points} \citep{FW1993}. To foreshadow, we show later that minimax designs can be obtained using a modification of Lloyd's algorithm. More recent applications of clustering in design include the Fast Flexible space-Filling (FFF) designs proposed by \cite{Lea2014}, which make use of hierarchical clustering to generate space-filling designs for computer experiments. A more in-depth discussion of these designs is provided in Section 2. 

The paper is outlined as follows. To better motivate the need for minimax designs, Section 2 begins with an overview of existing methods, then compares these methods with the proposed algorithm for a real-world example on air quality monitoring. Section 3 presents the new hybrid clustering algorithm for generating minimax designs, and provides some theoretical results on its correctedness and running time. Section 4 then outlines some numerical simulations comparing the proposed method with existing algorithms for a variety of design spaces. Section 5 introduces a new type of experimental design called \textit{minimax projection designs}, which are obtained by performing a simple refinement step on a minimax design. Finally, Section 6 discusses some future research directions.

\section{Background and motivation}
We begin by formally defining a minimax design:
\begin{definition}
\citep{Jea1990} Let $\mathcal{X} \subseteq \mathbb{R}^p$ be a desired design space. An $n$-point minimax design on $\mathcal{X}$ is defined as the optimal solution of
\be
\argmin_{\mathcal{D}_n \in \mathbb{D}_n} \sup_{\bm{x} \in \mathcal{X}} \|\bm{x} - Q(\bm{x},\mathcal{D}_n)\|,
\label{eq:minimax}
\ee
where $\mathbb{D}_n \equiv \{\{\bm{m}_i\}_{i=1}^n: \bm{m}_i \in \mathcal{X}\}$ is the set of all unordered $n$-tuples on $\mathcal{X}$, and $Q(\bm{x},\mathcal{D}_n) \equiv \argmin_{\bm{z} \in \mathcal{D}_n} \| \bm{x} - \bm{z} \|$ returns the nearest design point to $\bm{x}$ under norm $\| \cdot \|$.
\end{definition}
For the remainder of this paper, $\| \cdot \|$ is taken to be the Euclidean norm $\| \cdot \|_2$, although the proposed algorithm can easily be generalized to other norms.

This section begins by detailing the existing methods for generating minimax designs mentioned in the Introduction. A real-world application on air monitoring is then presented to motivate the importance of minimax designs in practice.

\subsection{Existing algorithms}

We first introduce the BIP algorithm in \cite{Tan2013}, which generates minimax designs on the finite design space $\mathcal{X} = \{\bm{y}_i\}_{i=1}^N$. Let $I_1, \cdots, I_N$ be binary decision variables, with $I_j = 1$ indicating point $j$ is included in the design and $I_j = 0$ otherwise. Also, let $\Omega_i$ denote the index set of points in $\mathcal{X}$ with (Euclidean) distance at most $S$. The BIP algorithm optimizes the following problem:
\be
\begin{aligned}
z(S) = \underset{I_1, \cdots, I_N}{\text{min}}  \sum_{j=1}^N I_j \quad \textrm{s.t.} \quad & \sum_{j \in \Omega_i} I_j \geq 1, \quad i = 1, \cdots, N, \quad I_j \in \{0,1\}, \quad j = 1, \cdots, N,\\
& d_{ij} = \| \bm{y}_i - \bm{y}_j\|_2, \quad \Omega_i = \{j \; : \; d_{ij} \leq S, j = 1, \cdots, N\}.
\end{aligned}
\label{eq:Tan}
\ee
In words, the optimization in \eqref{eq:Tan} chooses the smallest number of design points from $\mathcal{X}$, denoted as $z(S)$, needed to ensure all points in $\mathcal{X}$ are at most a distance of $S$ away from its nearest design point. The $n$-point minimax design can then be obtained by finding the smallest radius $S$ for which the optimal design size $z(S)$ satisfies $z(S) = n$. When the candidate points $\{\bm{y}_j\}_{j=1}^N$ are, in some sense, representative of a continuous design space, the design generated by BIP can be used to approximate the minimax design in \eqref{eq:minimax}.

Unfortunately, BIP has a major caveat which greatly limits its applicability in practice: the optimization in \eqref{eq:Tan} is computationally tractable only when the number of candidate points $N$ is small. For example, due to memory and time constraints, $N$ cannot exceed 1,000 for most desktop computers. In this sense, BIP is not only computationally demanding, but provides poor minimax designs when $p$ is large, since 1,000 points are insufficient for representing a high-dimensional space. This is illustrated in the simulations in Section 4.

Next, we discuss two types of clustering-based designs: principal points \citep{Flu1990} and FFF designs \citep{Lea2014}. Assume the design space $\mathcal{X}$ is convex and bounded, and let $U(\bm{X})$ denote the uniform distribution on $\mathcal{X}$. Just as minimax designs are defined as a minimizer of the minimax objective in \eqref{eq:minimax}, the \textit{principal points} of $U(\bm{X})$ are similarly defined as a minimizer of the integrated squared-error criterion:
\be
\argmin_{\mathcal{D}_n \in \mathbb{D}_n} \int_{\mathcal{X}} \|\bm{x} - Q(\bm{x},\mathcal{D}_n)\|_2^2 \; d \bm{x},
\label{eq:MSE}
\ee
where $\mathbb{D}_n$ and $Q(\bm{x},\mathcal{D}_n)$ are defined as in \eqref{eq:minimax}. In words, principal points aim to provide a uniform coverage of $\bm{X}$ by ensuring that, for a point uniformly sampled on $\mathcal{X}$, the expected squared-distance to its closest design point is minimized. Principal points are also known as \textit{minimum-MSE quantizers} in signal processing literature \citep{LBG1980}, and \textit{mse-rep-points} in quasi-Monte Carlo literature \citep{FW1993}.

To compute principal points, \cite{Flu1993} proposed the following two-step algorithm. First, generate a large random sample $\{\bm{y}_j\}_{j=1}^N \distas{i.i.d.} U(\bm{X})$, along with an initial design $\{\bm{m}_i\}_{i=1}^n \distas{i.i.d.} U(\bm{X})$. K-means clustering using Lloyd's algorithm \citep{Llo1957, Llo1982} is then performed with the large sample $\{\bm{y}_j\}_{j=1}^N$ as clustering data. In particular, Lloyd's algorithm iterates the following two updates until design points converge: (a) each sample point in $\{\bm{y}_j\}_{j=1}^N$ is first assigned to its closest design point; (b) each design point is then updated as the arithmetic mean of sample points assigned to it. The converged design is then taken as the principal points of $U(\mathcal{X})$. A similar algorithm is used in the popular Linde-Buzo-Gray (LBG) algorithm \citep{LBG1980} for generating minimum-MSE quantizers.

Justifying why such an algorithm provides locally optimal solutions of \eqref{eq:MSE} requires two lines of reasoning. First, using the random sample $\{\bm{y}_j\}_{j=1}^N$, the Monte Carlo approximation of \eqref{eq:MSE} becomes:

\small
\be
\begin{aligned}
& \underset{\boldsymbol{\gamma},\bm{m}_1, \cdots, \bm{m}_n}{\text{min}} \frac{1}{N}\sum_{i = 1}^n \sum_{j = 1}^N \gamma_{ij}\|\bm{y}_j - \bm{m}_i\|_2^2 \quad \text{s.t.} & \gamma_{ij} \in \{0,1\},  &\;  i = 1, \cdots, n,\; j = 1, \cdots, N;\\
& & \bm{m}_i \in \mathbb{R}^p, &\;  i = 1, \cdots, n; \; \sum_{i=1}^n \gamma_{ij} = 1, \; j = 1, \cdots, N. 
\label{eq:clust}
\end{aligned}
\ee
\normalsize
Here, $\boldsymbol{\gamma} = \{\gamma_{ij}\}$ is the set of binary decision variables, with $\gamma_{ij} = 1$ indicating the assignment of sample point $\bm{y}_j$ to design point $\bm{m}_i$. These binary variables serve the same role as $Q(\bm{x}, \mathcal{D}_n)$ in \eqref{eq:MSE}, namely, to assign each point in $\mathcal{X}$ to its closest design point. Likewise, the decision variables $\{\bm{m}_i\}_{i=1}^n$ correspond to the design optimization of $\mathcal{D}_n \in \mathbb{D}_n$ in \eqref{eq:MSE}. Second, the two updates in Lloyd's algorithm iteratively optimize the assignment variables $\{\gamma_{ij}\}$ and design points $\{\bm{m}_i\}$ in \eqref{eq:clust} respectively, while keeping other decision variables fixed. Specifically, by assigning each sample point $\bm{y}_j$ to its closest design point $\bm{m}_i$, the assignment variables $\{\gamma_{ij}\}$ in \eqref{eq:clust} are optimized for a fixed design $\{\bm{m}_i\}$. Similarly, by updating each design point $\bm{m}_i$ as the arithmetic mean of sample points assigned to it, the design $\{\bm{m}_i\}_{i=1}^n$ in \eqref{eq:clust} is optimized for fixed assignment variables. Iterating these updates until convergence therefore returns a locally optimal design for \eqref{eq:MSE}.

The FFF designs proposed by \cite{Lea2014} are of a similar flavor to principal points. These designs are generated by first obtaining a large sample $\{\bm{y}_j\}_{j=1}^N \distas{i.i.d.} U(\mathcal{X})$, conducting hierarchical clustering with Ward's minimum-variance criterion \citep{War1963} to form $n$ clusters of $\{\bm{y}_j\}_{j=1}^N$, then using cluster centroids as design points. The computation time of FFF designs can be shown to be $O(pN^2 \log N)$ \citep{Epp2000}, which suggests that, although these designs can be generated efficiently in high-dimensions for a fixed sample size $N$, its computation may be prohibitive when $N$ increases. To contrast, the proposed algorithm generates minimax designs efficiently both in high-dimensions and for large sample sizes.

In this paper, we compare the minimax performance of BIP designs, principal points and FFF designs to the designs generated by the proposed method. To reiterate, while the latter two designs are not intended for minimax use, they are included to provide a benchmark for our algorithm, and to show that such designs indeed provide decent minimax performance.

\subsection{Motivating example: Air quality monitoring}
\begin{figure}[!t]
\centering
\includegraphics[width=\textwidth]{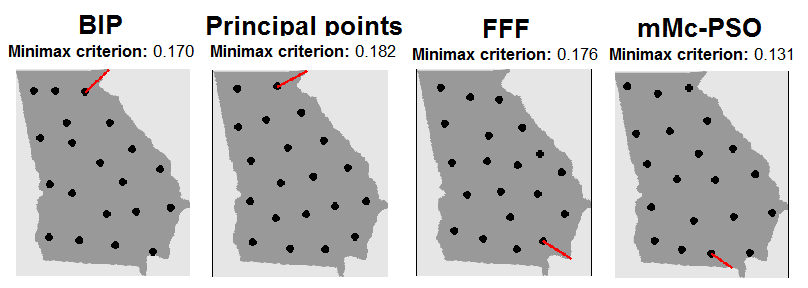}
\caption{Four different 20-point designs for the state of Georgia. The red line on each plot connects the point in Georgia furthest from the design to its nearest design point, with its length equal to the minimax criterion of the design. Of these four designs, the new method \textsc{mMc-PSO} provides the best minimax design.}
\label{fig:GAex}
\end{figure}

To motivate the use of minimax designs in real-world situations, consider the problem of air quality monitoring in the state of Georgia. With wildfire occurrences and air pollution levels on the rise in many parts of the United States \citep{Bea2008}, there is an increasing need for precise air quality monitoring, both for supporting warning systems and for guiding public health and policy decisions. To this end, many states have adopted the Ambient Monitoring Program (AMP), which requires hourly reporting of concentration levels for six key air pollutants. Unfortunately, only a small number of monitoring stations can be set-up for each state, since the building and maintenance of these stations can be very expensive. As a result, there are only 30 such stations situated in the state of Georgia \citep{Ose2016}. A key problem then is to allocate these limited stations in such a way that each part of the state is covered sufficiently well by a station. The optimal allocation scheme, by definition, is that provided by a minimax design.

Figure \ref{fig:GAex} plots the 20-point designs generated by the three existing methods: BIP, principal points and FFF, along with the design generated by the proposed algorithm \textsc{mMc-PSO}. The red line on each plot connects the point in Georgia furthest from the design to its nearest design point. Note that the minimax criterion in \eqref{eq:minimax} (reported at the top of each plot) corresponds to the length of this line. Two key observations can be made here. First, principal points and \textsc{mMc-PSO} appear to provide the best visual uniformity of the four methods, whereas the design generated by BIP appears to be visually non-uniform. Second, \textsc{mMc-PSO} provides the lowest minimax distance of the four methods, which illustrates the improvement that the proposed method offers over existing methods. We show that this improvement holds for a wide range of design regions in Section 4.

\section{Methodology}
In this section, we first present the \textit{minimax clustering} algorithm as a generalization of Lloyd's algorithm, then establish theoretical results for the correctedness and running time for the proposed method. Finally, we introduce a global optimization modification for minimax clustering, which allows near-optimal minimax designs to be generated. 

\subsection{Minimax clustering}

To begin, we introduce a new type of center for a finite set of points:
\begin{definition}
For a finite set of $m$ points $\mathcal{Z} = \{\bm{z}_i\}_{i=1}^m \subseteq \mathbb{R}^p$, its $C_q$-center is defined as:
\be
\argmin_{\bm{z}\in\mathbb{R}^p} D_{q}(\bm{z}; \mathcal{Z}) \equiv \argmin_{\bm{z}\in\mathbb{R}^p} \frac{1}{mq} \sum_{i = 1}^m \| \bm{z} - \bm{z}_i\|_2^q.
\label{eq:ckcent}
\ee
\end{definition}
$C_q$-centers can be seen as Fr\'echet means \citep{NB2013}, which are of the form $\argmin_{\bm{z}\in\mathbb{R}^p} \sum_{i=1}^m w_i d(\bm{z},\bm{z}_i)$, with weights $w_i = 1/(mq)$ and distance function $d(\bm{x},\bm{y}) = \| \bm{x} - \bm{y}\|_2^q$. With $q = 2$, the $C_q$-center becomes the arithmetic mean, used for updating cluster centers in Lloyd's algorithm. More importantly, as $q \rightarrow \infty$, the $C_{q}$-center returns the point which \textit{minimizes} the \textit{maximum} distance between it and a point in $\mathcal{Z}$. To foreshadow, $C_\infty$-centers will be used in place of arithmetic means in the proposed clustering scheme.

The intuition for minimax clustering can then be presented by direct analogy to principal points. Consider the minimax objective in \eqref{eq:minimax}, and note that for sufficiently large choices of $q > 0$, this objective can be approximated as:
\begin{align}
\argmin_{\mathcal{D}_n \in \mathbb{D}_n} \left( \int_{\mathcal{X}} \| \bm{x} - Q(\bm{x},\mathcal{D}_n) \|_2^q \; d\bm{x} \right)^{1/q},\label{eq:minimaxk}
\end{align}
In practice, $q$ should be large enough to provide a good approximation of \eqref{eq:minimax}, yet small enough to avoid numerical instability. The choice of $q$ is discussed further in Section 3.2.1.

The similarities between the approximation \eqref{eq:minimaxk} and the integrated squared-error \eqref{eq:MSE} allows for a modification of Lloyd's algorithm to generate minimax designs. First, generate a large sample $\{\bm{y}_j\}_{j=1}^N \distas{i.i.d.} U(\mathcal{X})$, along with initial cluster centers $\{\bm{m}_i\}_{i=1}^n \distas{i.i.d.} U(\mathcal{X})$. The Monte Carlo approximation of \eqref{eq:minimaxk} becomes:
\small
\be
\begin{aligned}
& \underset{\boldsymbol{\gamma},\bm{m}_1, \cdots, \bm{m}_n}{\text{min}} \frac{1}{N}\sum_{i = 1}^n \sum_{j = 1}^N \gamma_{ij}\|\bm{y}_j - \bm{m}_i\|_q^2 \quad \text{s.t.} & \gamma_{ij} \in \{0,1\},  &\;  i = 1, \cdots, n,\; j = 1, \cdots, N;\\
& & \bm{m}_i \in \mathbb{R}^p, &\;  i = 1, \cdots, n; \; \sum_{i=1}^n \gamma_{ij} = 1, \; j = 1, \cdots, N. 
\label{eq:clustk}
\end{aligned}
\ee
\normalsize
where $\boldsymbol{\gamma} = \{\gamma_{ij}\}$ is again the set of binary assignment variables, and $\{\bm{m}_i\}_{i=1}^n$ the set of design points. Minimax clustering then iteratively applies the following two updates until design points converge: (a) each sample point in $\{\bm{y}_j\}_{j=1}^N$ is first assigned to its closest design point, which optimizes the assignment variables $\{\gamma_{ij}\}$ in \eqref{eq:clustk} for a fixed design $\{\bm{m}_i\}$; (b) each design point is then updated as the $C^{(q)}$-center of points assigned to it, which optimizes the design $\{\bm{m}_i\}_{i=1}^n$ in \eqref{eq:clustk} for fixed assignments. By iterating these two updates until convergence, one should obtain a locally-optimal minimax design. The above procedure, which we call \textit{minimax clustering} (or \textsc{mMc} for short), is summarized in Algorithm \ref{alg:mMc}.

\begin{algorithm}[t]
\caption{Minimax clustering}
\label{alg:mMc}
\begin{algorithmic}[1]
\small
\Function{mMc}{$\{\bm{m}_i\}_{i=1}^n,N,q,t_{mMc},\epsilon_{in}$} \Comment{$\{\bm{m}_i\}_{i=1}^n$ - initial design, $t_{mMc}$ - max. iterations}
\stb Initialize $\{\bm{y}_j\}_{j=1}^N$ using a Sobol' sequence
\Repeat
\stb For $j = 1, \cdots, N$, assign $\bm{y}_j$ to its closest design point in Euclidean norm.
\stb For $i=1, \cdots, n$, update $\bm{m}_i \leftarrow C_q$\textsc{-AGD}$(\mathcal{Z}_i,q,\epsilon_{in})$, where $\mathcal{Z}_i$ is the set of points assigned to $\bm{m}_i$
\stb $t \leftarrow t + 1$.
\Until{design points converge \textbf{OR} $t \geq t_{mMc}$.}
\stb \Return converged design $\{\bm{m}_i\}_{i=1}^n$.
\EndFunction
\end{algorithmic}
\end{algorithm}

In our implementation, deterministic low-discrepancy sequences \citep{Nie1992} are used in place of random samples for $\{\bm{y}_j\}_{j=1}^N$, since such sequences provide a better approximation of integrals compared to Monte Carlo methods. Assume for now that the design space $\mathcal{X}$ is $[0,1]^p$, the unit hypercube in $\mathbb{R}^p$. We employ a specific type of low-discrepancy sequence in Algorithm \ref{alg:mMc} called a \textit{Sobol' sequence} \citep{Sob1967}, which can be generated efficiently using the function \texttt{sobol} in the R package \texttt{randtoolbox} \citep{DS2013}. Section 4.2 provides a brief discussion on low-discrepancy sequences for general design spaces.

\subsection{Convergence results}
The above discussion still leaves two questions unanswered. First, how can $C_q$-centers computed efficiently? Second, does minimax clustering indeed converge in finite iterations to a local optimum, and if so, at what rate? These concerns are addressed in this subsection.

Since the discussion below is quite technical, readers interested in the hybridization of \textsc{mMc} with particle swarm should skip to Section \ref{sec:pso}. Some background readings on convex programming (e.g., \citealp{BN2001} and \citealp{BV2004}) may also be useful for understanding the developments in this subsection. For brevity, proofs are deferred to supplementary materials.

\subsubsection{Computing $C_q$-centers}
We first present an algorithm for computing $C_q$-centers, and prove that this algorithm converges quickly even when the number of points $m$ or dimension $p$ become large. The following theorem shows that the objective $D_{q}(\bm{z}; \mathcal{Z})$ in \eqref{eq:ckcent} is strictly convex, and that the $C_q$-center of $\mathcal{Z}$ is unique and contained in the convex hull of $\mathcal{Z}$, defined as $\text{conv}(\mathcal{Z}) = \{\bm{z}= \sum_{i=1}^m \alpha_i \bm{z}_i : \alpha_i \geq 0, \sum_{i=1}^m \alpha_i = 1\}$.

\begin{theorem}
Let $\mathcal{Z} = \{\bm{z}_i\}_{i=1}^m$ and let $q \geq 2$. Then $D_{q}(\bm{z}; \mathcal{Z})$ is strictly convex in $\bm{z}$. Moreover, the $C_q$-center $C_q(\mathcal{Z})$ in \eqref{eq:ckcent} is unique, and contained in $\text{\normalfont{conv}}(\mathcal{Z})$.
\label{thm:uni}
\end{theorem}

Next, recall that a function $h: \mathbb{R}^p \rightarrow \mathbb{R}$ is $\beta$-\textit{Lipschitz smooth} (or simply $\beta$-\textit{smooth}) if:
\[\|\nabla h(\bm{z}) - \nabla h(\bm{z}')\|_2 \leq \beta \|\bm{z} - \bm{z}'\|_2,\]
where $\nabla h$ is the gradient of $h$. Likewise, $h$ is $\mu$\textit{-strongly convex} if:
\[(\nabla h(\bm{z}) - \nabla h(\bm{z}'))^T (\bm{z} - \bm{z}') \geq \mu \|\bm{z} - \bm{z}'\|_2.\]
We show next that, for some specified $\bar{\beta} > 0$ and $\bar{\mu} > 0$, the objective function $D_{q}( \bm{z}; \mathcal{Z})$ is $\bar{\beta}$-smooth and $\bar{\mu}$-strongly convex.

\begin{theorem}
For $q \geq 4$, $D_{q}( \bm{z}; \mathcal{Z})$ is $\bar{\beta}$-smooth and $\bar{\mu}$-strongly convex for $\bm{z} \in \text{conv}(\mathcal{Z})$, where:
\begin{equation}
\bar{\beta} = (q-1)(q-2) \max_{j=1, \cdots, m} D_{q-2}( \bm{z}_j; \mathcal{Z}) > 0 \quad \text{and} \quad \bar{\mu} = (q-2)D_{q-2}( C_{q-2}(\mathcal{Z}); \mathcal{Z})> 0.
\label{eq:bars}
\end{equation}
\label{thm:lip}
\end{theorem}
\vspace{-0.3cm}

\begin{algorithm}[t]
\caption{Computing $C_q$-centers}
\label{alg:cq}
\begin{algorithmic}[1]
\small
\Function{$C_q$-AGD}{$\{\bm{z}_i\}_{i=1}^m$, $q$, $\epsilon_{in}$} \Comment {$\epsilon_{in}$ - desired tolerance}
\stb Set $t = 1$ and initialize starting points $\bm{z}^{[1]} \leftarrow \frac{1}{m} \sum_{i=1}^m \bm{z}_i$, $\bm{u}^{[1]} \leftarrow \frac{1}{m} \sum_{i=1}^m \bm{z}_i$.
\stb Initialize the sequences $\{\lambda_t\}_{t=0}^\infty$ and $\{\gamma_t\}_{t=1}^\infty$ from \eqref{eq:seq}.
\stb Compute the Lipschitz constant $\bar{\beta}$ in \eqref{eq:bars}.
\While{$\|\bm{z}^{[t]} - \bm{z}^{[t-1]}\|_2 < \epsilon_{in}$}
\stb Update $\bm{u}^{[t+1]} \leftarrow \bm{z}^{[t]} - \frac{1}{\bar{\beta}} \left( \frac{1}{m} \sum_{i=1}^m \|\bm{z}^{[t]}-\bm{z}_i\|_2^{q-2} (\bm{z}^{[t]}-\bm{z}_i) \right)$.
\stb Update $\bm{z}^{[t+1]} \leftarrow (1 - \gamma_t) \bm{u}^{[t+1]} + \gamma_t \bm{u}^{[t]}$.
\stb $t \leftarrow t + 1$.
\EndWhile
\stb \Return $\bm{z}^{[t]}$.
\EndFunction
\end{algorithmic}
\end{algorithm}

The $\bar{\beta}$-smoothness and $\bar{\mu}$-strong convexity in Theorem \ref{thm:lip} allow us to employ a quick convex optimization technique called \textit{accelerated gradient descent} \citep{Nes1983}, or AGD, to compute $C_q$-centers. The implementation of AGD is straightforward. Suppose $h: \mathbb{R}^p \rightarrow \mathbb{R}$, the desired objective to minimize, is twice-differentiable, convex and $\beta$-smooth. Let $\bm{u}^{[t]}  \in \mathbb{R}^p$ be the $t$-th solution iterate, and let $\bm{z}^{[t]} \in \mathbb{R}^p$ be an intermediate vector. Also, define the sequences $\{\lambda_t\}_{t=0}^\infty$ and $\{\gamma_t\}_{t=1}^\infty$ by the recursion equations:
\begin{equation}
\lambda_0 = 0, \quad \lambda_{t} = \frac{1+\sqrt{1+4\lambda_{t-1}^2}}{2}, \quad \gamma_{t} = \frac{1-\lambda_{t}}{\lambda_{t+1}} \quad \text{for $t = 1, 2, 3, \cdots$}.
\label{eq:seq}
\end{equation}
AGD then iterates the following two updates until the solution sequence $\{\bm{u}^{[t]}\}_{t=1}^\infty$ converges:
\begin{equation}
\bm{u}^{[t+1]} \leftarrow \bm{z}^{[t]} - \frac{1}{\beta}\nabla h(\bm{z}^{[t]}), \quad \bm{z}^{[t+1]} \leftarrow (1-\gamma_t) \bm{u}^{[t+1]} + \gamma_t \bm{u}^{[t]}.
\label{eq:agd}
\end{equation}
A direct application of AGD for the optimization in \eqref{eq:ckcent} is provided in Algorithm \ref{alg:cq}.

One may perhaps ask why this accelerated scheme is preferred over traditional line-search methods (see, e.g., \citealp{NW2006}), in which the solution sequence $\{\bm{u}^{[t]}\}_{t=1}^\infty$ is updated by the line-search optimization:
\begin{equation}
\bm{u}^{[t+1]} = \bm{u}^{[t]} - \eta_t \nabla h(\bm{u}^{[t]}), \quad \eta_t = \argmin_{\eta > 0} h(\bm{u}^{[t]} - \eta \nabla h(\bm{u}^{[t]})).
\label{eq:linesearch}
\end{equation}
In other words, for a given iterate $\bm{u}^{[t]}$, the next iterate $\bm{u}^{[t+1]}$ in line-search methods is obtained by searching for the optimal step-size $\eta_t$ to move along the direction of its negative gradient $-\nabla h(\bm{u}^{[t]})$. The advantages of AGD are two-fold. First, AGD exploits the $\beta$-smoothness and $\mu$-convexity of \eqref{eq:ckcent} to achieve an optimal rate of convergence among gradient-based optimization methods \citep{Nes2013}. Second, the step-size optimization in \eqref{eq:linesearch} requires multiple evaluations of the objective $h$ and its gradient $\nabla h$. Since the evaluation of both $D_{q}( \bm{z}; \mathcal{Z})$ and $\nabla D_{q}( \bm{z}; \mathcal{Z})$ require $O(mp)$ work, such evaluations become prohibitively expensive to compute when either the number of points $m$ or dimension $p$ are large. AGD avoids this problem by replacing the optimized step-size $\eta_t$ with a fixed stepsize $1/\bar{\beta}$.

Using Theorem \ref{thm:lip}, the correctedness and running time of Algorithm \ref{alg:cq} can be established.

\begin{corollary}
For $\mathcal{Z} = \{\bm{z}_i\}_{i=1}^m$ and $q \geq 4$, consider the sequence of solutions $\{\bm{z}^{[t]}\}_{t=1}^\infty$ from Algorithm \ref{alg:cq}. To guarantee an $\epsilon_{in}$-accuracy for the objective in \eqref{eq:ckcent}, i.e., $|D_{q}( \bm{z}^{[t]} ; \mathcal{Z}) - D_{q}( C_q(\mathcal{Z}) ; \mathcal{Z})| < \epsilon_{in}$, the computation work required is:
\begin{equation}
O\left( mp \sqrt{(q-1)\kappa_{q-2}(\mathcal{Z})} \log \frac{1}{\epsilon_{in}}\right), \; \text{where } \; \kappa_{q}(\mathcal{Z}) = \frac{\max_{j=1, \cdots, m} D_{q}( \bm{z}_j; \mathcal{Z})}{D_{q}( C_{q}(\mathcal{Z}); \mathcal{Z})}
\label{eq:itupp}
\end{equation}
is the ratio of maximum and minimum values of $D_{q}( \bm{z} ; \mathcal{Z})$ for $\bm{z} \in \text{conv}(\mathcal{Z})$.
\label{corr:cqconv}
\end{corollary}

Several illuminating observations can be made from this corollary. First, considering only the error tolerance $\epsilon_{in}$, the computational work required for AGD to achieve $\epsilon_{in}$-accuracy is $O(\log(1/\epsilon_{in}))$, which is sizably smaller than the $O(1/\epsilon_{in})$ work needed for standard line-search methods \citep{NW2006}. Hence, Algorithm \ref{alg:cq} not only avoids multiple evaluations of the objective and gradient, but also converges with fewer iterations compared to line-search methods. Second, the bound in \eqref{eq:itupp} grows on the order of $\sqrt{q}$, meaning Algorithm \ref{alg:cq} takes longer to terminate as $q$ grows larger. This illustrates the trade-off between performance and accuracy: a larger value of $q$ ensures a better approximation of the minimax criterion \eqref{eq:minimaxk}, but requires longer time to compute. In our simulations, $q = 10$ appears to provide a good compromise in this trade-off. Lastly, the bound in \eqref{eq:itupp} grows as $\kappa_{q-2}(\mathcal{Z})$ increases, meaning $C_q$-centers may take longer to compute when points in $\mathcal{Z}$ are more scattered.

\subsubsection{Correctedness and running time of minimax clustering}

The correctedness and running time of minimax clustering can then be established by direct analogy to that for Lloyd's algorithm. This is formally demonstrated below.
\begin{theorem}
Algorithm \ref{alg:mMc} terminates after at most $N^n$ iterations. Moreover, assuming $n \leq {N}^{1/2}$, each iteration of the loop in Algorithm \ref{alg:mMc} requires $O\left( N^{3/2} p\sqrt{q-1} \log \frac{1}{\epsilon_{in}}\right)$ work, where $\epsilon_{in}$ is the inner tolerance in Corollary \ref{corr:cqconv}. Lastly, when $C_q$-center updates in \eqref{eq:ckcent} are exact, Algorithm \ref{alg:mMc} also returns a locally optimal design for \eqref{eq:clustk}.
\label{thm:mMconv}
\end{theorem}

Unfortunately, it is difficult to establish a bound on the number of iterations required for termination of Algorithm \ref{alg:mMc}, since there is still a gap between theory and practice for the same problem in Lloyd's algorithm. Theoretical work (\citealp{Aea2009, Bho2009}) suggests that in the worst-case, the number of iterations can grow rapidly in the number of clustering points $N$. However, in practice, Lloyd's algorithm nearly always terminates after several iterations, leading many practitioners (see, e.g., \citealp{AS2007}) to evaluate total running time by the running time of one iteration. From our simulations, Algorithm \ref{alg:mMc} also converges after a small number of iterations, so we similarly use the single-iteration time in Theorem \ref{thm:mMconv} to measure for total running time of minimax clustering. 

In this light, the running time of Theorem \ref{thm:mMconv} illustrates two computational advantages of minimax clustering. First, since this time is linear in $p$, minimax clustering can be performed efficiently in high-dimensions, which is similar to what is observed for FFF designs in Section 2.1. Furthermore, the running time of minimax clustering grows at a rate of $N^{3/2}$, which is much faster than the $O(N^2 \log N)$ work for FFF designs. Hence, a larger number of approximating points $N$ can be used in minimax clustering, suggesting that the proposed method provides higher quality minimax designs when $\mathcal{X}$ is high-dimensional. As we see later in Section 4, this is indeed the case.

\subsection{Minimax clustering with particle swarm optimization}
\label{sec:pso}
Due to its greedy nature, Lloyd's algorithm has two drawbacks: it is sensitive to choices of initial cluster centers, and may return a locally optimal design which is far from the global design \citep{Jai2010}. Since minimax clustering employs the same greedy steps, it suffers from the same downfalls. A simple but computationally expensive remedy is to perform Lloyd's algorithm multiple times with different initial centers, then pick the solution with the smallest criterion in \eqref{eq:clust}. More elaborate methods requiring less computation include kernel k-means \citep{TL2009}, sequential k-means \citep{LVV2003}, and combining k-means with particle swarm optimization \citep{VE2003}. To retain the iterative nature of Algorithm \ref{alg:mMc}, we adopt the latter hybrid approach for global optimization of minimax clustering.

\begin{algorithm}[!t]
\caption{Minimax clustering with PSO}
\label{alg:mMcPSO}
\small
\begin{algorithmic}[1]
\Function{\textsc{mMc-PSO}}{$n,N,q,s,t_{mMc},t_{pp},\epsilon_{in}$}
\stb Generate $\{\bm{y}_j\}_{j=1}^N$ using a Sobol' sequence and initial design particles $\mathcal{D}_{k} = \{\bm{m}^{k}_i\}_{i=1}^n, k = 1, \cdots, s$ using scrambled Sobol' sequences.
\stb Define $h_q$ as the objective in \eqref{eq:clustk}, and $h$ as the minimax criterion in \eqref{eq:minimax} with $\mathcal{X} = \{\bm{y}_j\}_{j=1}^N$.
\stb \textit{Minimax clustering PSO:} Initialize  local-best designs $\mathcal{L}_k \leftarrow \mathcal{D}_{k}, k = 1, \cdots, s$, and global-best design $\mathcal{G} \leftarrow \argmin_{\mathcal{D}_{k}} h_q( \mathcal{D}_{k})$. Set initial velocities $\bm{v}_k \leftarrow \bm{0}, k = 1, \cdots, s$.
\For{$t = 1, \cdots, t_{mMc}$} \Comment{$t_{mMc}$ - max. PSO iterations}
\For{$k = 1, \cdots, s$} \Comment{For each design particle...}
\stb $\mathcal{D}_k \leftarrow$ \textsc{mMc}($\mathcal{D}_{k},N,q,1,\epsilon_{in}$) \Comment{One step of minimax clustering}
\stb $\bm{v}_k \leftarrow w \bm{v}_k + c_1 \bm{r}_1 (\mathcal{L}_k - \mathcal{D}_k)  + c_2 \bm{r}_2 (\mathcal{G} - \mathcal{D}_k), \; \bm{r}_1,\bm{r}_2 \distas{i.i.d.} U[0,1]^{np}$ \Comment{Update vel.}
\stb $\mathcal{D}_k \leftarrow \mathcal{D}_k + \bm{v}_k$ \Comment{Move particle towards best positions}
\If{$h_q(\mathcal{D}_k) < h_q(\mathcal{L}_k)$} $\mathcal{L}_k \leftarrow \mathcal{D}_k$ \Comment{Update local-best designs}
\EndIf
\If{$h_q(\mathcal{D}_k) < h_q(\mathcal{G})$} $\mathcal{G} \leftarrow \mathcal{D}_k$ \Comment{Update global-best design}
\EndIf
\EndFor
\EndFor
\stb \textit{Post-processing:} Reset global-best design $\mathcal{G} \leftarrow \argmin_{\mathcal{D}_{k}} h( \mathcal{D}_{k})$ and velocities $\bm{v}_k \leftarrow \bm{0}$.
\For{$t = 1, \cdots, t_{pp}$} \Comment{$t_{pp}$ - max. post-proc. iterations}
\For{$k = 1, \cdots, s$} \Comment{For each design particle...}
\stb $\bm{v}_k \leftarrow w \bm{v}_k + c_1 \bm{r}_1 (\mathcal{L}_k - \mathcal{D}_k)  + c_2 \bm{r}_2 (\mathcal{G} - \mathcal{D}_k), \; \bm{r}_1,\bm{r}_2 \distas{i.i.d.} U[0,1]^{np}$ \Comment{Update vel.}
\stb $\mathcal{D}_k \leftarrow \mathcal{D}_k + \bm{v}_k$ \Comment{Move particle towards best positions}
\If{$h(\mathcal{D}_k) < h(\mathcal{L}_k)$} $\mathcal{L}_k \leftarrow \mathcal{D}_k$ \Comment{Update local-best designs}
\EndIf
\If{$h(\mathcal{D}_k) < h(\mathcal{G})$} $\mathcal{G} \leftarrow \mathcal{D}_k$ \Comment{Update global-best design}
\EndIf
\EndFor
\EndFor
\stb \Return global-best design $\mathcal{G}$.
\EndFunction
\end{algorithmic}
\end{algorithm}

Particle swarm optimization \citep{EK1995}, or PSO for short, is a stochastic, derivative-free algorithm for global minimization of a general function $h$. This algorithm can be described as follows. First, a representative set of $s$ feasible solutions, or a \textit{swarm} of \textit{particles}, is chosen. Each particle is then guided towards the solution with lowest objective encountered along its own path (called the \textit{local-best solution}), as well as the solution with lowest objective over the entire swarm (called the \textit{global-best solution}). In this sense, PSO mimics the behavior of a bird flock searching for food: each bird naturally flies towards the closest position to a food source explored by the flock, but is also guided by the closest position explored along its own flight. When the optimization problem at hand has some desirable structure, PSO can be combined (or \textit{hybridized}) with other algorithms to provide quicker convergence. We therefore propose a hybridization scheme below which combines PSO with the minimax clustering algorithm \textsc{mMc}.

The details are as follows. First, generate the set of approximating points $\{\bm{y}_j\}_{j=1}^N$ using a Sobol' sequence, and generate the $s$ initial designs (forming the \textit{particle swarm}) using scrambled Sobol' sequences \citep{Owe1995}. In non-technical terms, these scrambled sequences provide different initial designs in the swarm, with each retaining its low-discrepancy property. Next, repeat the following steps:
\bi
\setlength\itemsep{0em}
\item For each design particle, do one iteration of minimax clustering.
\item Move each design particle towards to its local-best and global-best designs.
\item Update the local-best and global-best designs for the desired objective in \eqref{eq:clustk}.
\ei
Finally, as a post-processing step, the general version of PSO described previously is applied to the minimax objective \eqref{eq:minimax}, with $\mathcal{X}$ approximated by $\{\bm{y}_j\}_{j=1}^N$. The above procedure, which we call \textsc{mMc-PSO}, is detailed in Algorithm \ref{alg:mMcPSO}. \textsc{mMc-PSO} will be used to generate the minimax designs in our simulations later.

Three parameters are used to control the PSO behavior of \textsc{mMc-PSO}: $c_1$ and $c_2$, which account for the velocities at which each particle drifts towards its local-best and global-best solutions respectively, and $w$, which controls each particle's momentum from one iteration to the next. For the PSO of Lloyd's algorithm proposed by \cite{VE2003}, the authors recommend the setting of $w = 0.72$ and $c_1 = c_2 = 1.49$, which can be shown to provide quick empirical convergence. Since this variant is similar to \textsc{mMc-PSO}, we adopt the same choices here. Other settings have also been tested, but we found this setting to provide the best minimax performance.

\begin{figure}[t]
\centering
\includegraphics[scale=0.60]{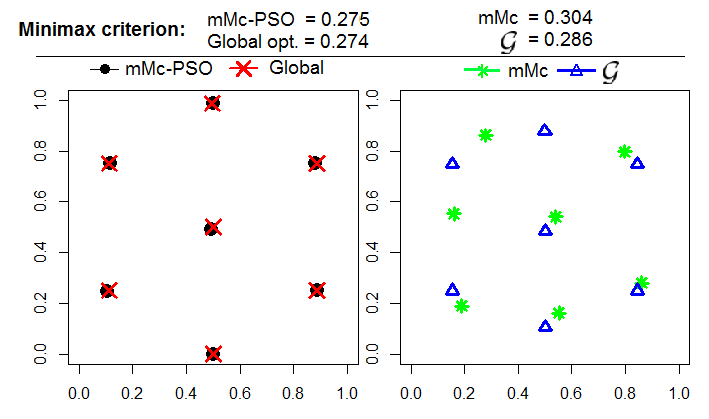}
\caption{(Left) The 7-point design using \textsc{mMc-PSO} and the global minimax design in \cite{Jea1990}. Since these designs are nearly identical, this demonstrates the near-global minimax performance of \textsc{mMc-PSO}. (Right) The 7-point design using \textsc{mMc} and the global-best design $\mathcal{G}$ in \textsc{mMc-PSO} before post-processing. The reduction in minimax distance for the latter design highlights the need for PSO.}
\label{fig:Jea}
\end{figure}

To illustrate the ability for \textsc{mMc-PSO} to generate near-global minimax designs, we compare the 7-point design for $p=2$ from \textsc{mMc-PSO} with the global minimax design in \cite{Jea1990}. Here, $N = 10^5$ approximating points are used, along with $s = 10$ PSO particles. The maximum iteration counts are set at $t_{mMc}=300$ and $t_{pp}=300$. The left plot in Figure \ref{fig:Jea} compares the design generated by \textsc{mMc-PSO} with the global minimax design. Visually, these two designs are nearly identical. Objective-wise, the minimax distance \eqref{eq:minimax} for \textsc{mMc-PSO} is within 0.001 of the global minimum, suggesting that the proposed algorithm indeed provides near-global optimization of \eqref{eq:minimax}. Similar results also hold for the remaining designs in \cite{Jea1990}, but these are not reported for brevity.

The right plot in Figure \ref{fig:Jea}, which outlines the 7-point design from Algorithm \ref{alg:mMc} (minimax clustering \textit{without} PSO) and the global-best design $\mathcal{G}$ in \textsc{mMc-PSO} \textit{before} post-processing, highlights the effectiveness of both PSO and post-processing. From this figure, $\mathcal{G}$ clearly gives a better approximation of the global design than mMc, both visually and criterion-wise, which suggests that the proposed PSO for minimax clustering is indeed effective. However, there is one glaring problem with $\mathcal{G}$: design points are pushed away from the boundaries of $[0,1]^2$, whereas two design points can be found on the top and bottom boundaries for the global minimax design. The post-processing step on $\mathcal{G}$, which performs PSO directly on the minimax criterion \eqref{eq:minimax}, allows design points to move towards their globally optimal positions on design boundaries.

\section{Numerical simulations}

In this section, we compare the minimax performance of designs using \textsc{mMc-PSO} with the existing methods in Section 2.1.  The comparison is first made on the unit hypercube $[0,1]^p$, then on the unit simplex and ball. This section concludes by returning to the original motivating example on air quality monitoring.

\subsection{Minimax designs on $[0,1]^p$}

\begin{figure}[!t]
\centering
\includegraphics[scale=0.60]{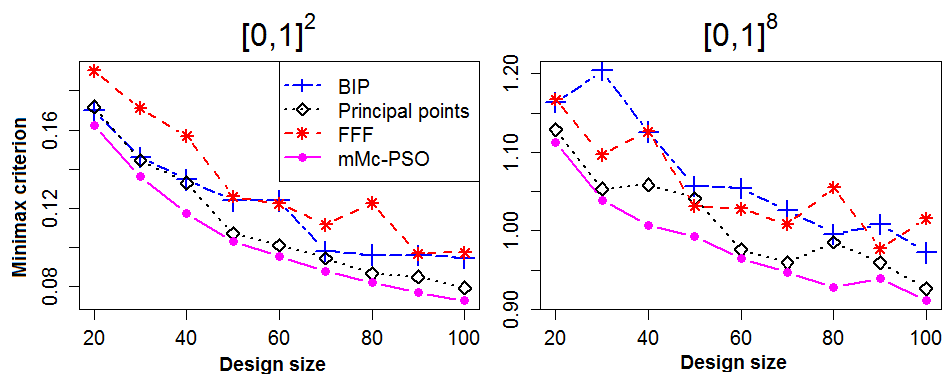}
\caption{Minimax criterion for various design sizes on $[0,1]^2$ and $[0,1]^8$. Designs generated by \textsc{mMc-PSO} consistently give the lowest minimax distance for all design sizes.}
\label{fig:minimaxu}
\end{figure}
\begin{figure}[t]
\centering
\includegraphics[width=\textwidth]{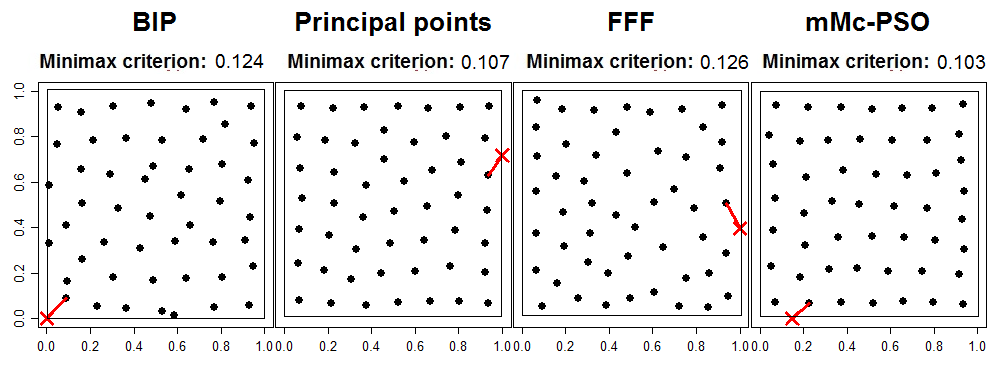}
\caption{Four different 50-point designs for $[0,1]^2$. The red line on each plot connects the point in $[0,1]^2$ furthest from the design (marked by `x') to its nearest design point, with its length equal to the minimax criterion. The proposed method \textsc{mMc-PSO} again provides the best minimax design.}
\label{fig:50u}
\end{figure}

We first illustrate the minimax performance and computation time of \textsc{mMc-PSO} on the unit hypercube $[0,1]^p$ in $p=2, 4, 6$ and $8$ dimensions. For brevity, only results for $p=2$ and $p=8$ are reported here, with additional results deferred to supplementary materials. The simulation settings are as follows. For \textsc{mMc-PSO}, we generate $n=20, 30, \cdots, 100$-point designs using $s=10$ PSO particles with $N=10^5$ approximating points. The maximum iterations in Algorithm \ref{alg:mMcPSO} are set at $t_{mMc}=500$ and $t_{pp}=250$. Our implementation of \textsc{mMc-PSO} is written in C++, and is available in the R package \texttt{minimaxdesign} \citep{Mak2016} in CRAN. For principal points, $N = 10^5$ approximating points are also used to provide a fair comparison with \textsc{mMc-PSO}. Lastly, for BIP, designs of the same sizes are generated with the candidate set taken from the first 1,000 points of the Sobol' sequence. FFF designs are also generated from JMP 12 using the cluster centers option.

For each design, Figure \ref{fig:minimaxu} plots the minimax criterion \eqref{eq:minimax} with $\mathcal{X} = [0,1]^p$ approximated by the first $10^7$ points from the Sobol' sequence. For $p = 2$, designs generated using \textsc{mMc-PSO} have the lowest minimax distance of the four methods for all design sizes $n$, which shows the proposed method indeed provides better minimax designs compared to existing methods. FFF designs, on the other hand, have the largest minimax distance for nearly all design sizes. Surprisingly, designs generated using BIP also have large minimax distances, suggesting that a candidate set of 1,000 design points is insufficient for representing the unit hypercube even in 2 dimensions. On the other hand, even though principal points provide relatively higher minimax distance compared to \textsc{mMc-PSO}, it is consistently better than BIP or FFF. Hence, although principal points are not intended for minimax use, the minimax performance of these designs can be quite good. From Figure \ref{fig:50u}, which plots the 50-point designs for the four methods, principal points and \textsc{mMc-PSO} also enjoy a more visually uniform coverage of $[0,1]^2$ compared to FFF and BIP.

From the right plot of Figure \ref{fig:minimaxu}, similar results hold for $p=8$ as well. \textsc{mMc-PSO} again provides the best minimax designs, with the improvement gap in minimax distance greater than that for $p=2$. This suggests that \textsc{mMc-PSO} provides an increasing improvement over existing methods as dimension $p$ increases. A contributing factor is the ability for \textsc{mMc-PSO} to manipulate a larger number of approximating points $N$ compared to FFF or BIP, an observation which was made in Section 3.2.2. This then allows the proposed algorithm to provide better minimax designs in high-dimensions.

\begin{figure}[t]
\includegraphics[scale=0.57]{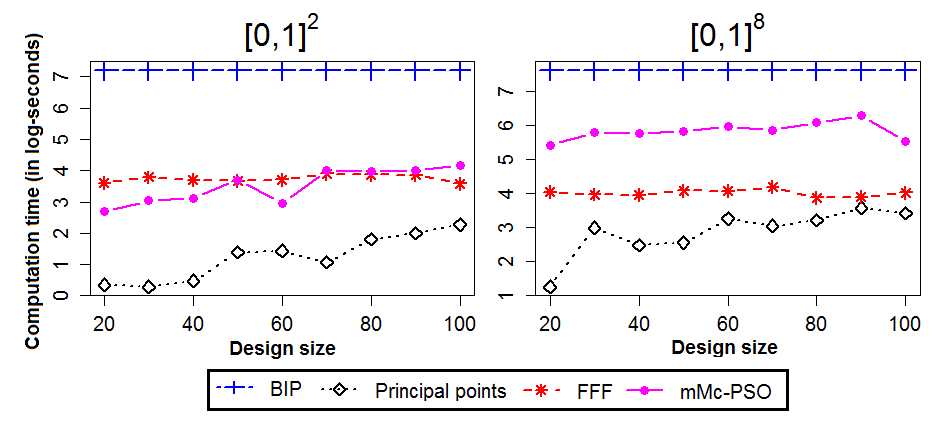}
\caption{Time (in log-seconds) required for generating designs on $[0,1]^p$. The computation times for \textsc{mMc-PSO} are slightly higher than principal points and FFF, but lower than BIP.}
\label{fig:time}
\end{figure}

For computation time, Figure \ref{fig:time} plots the time (in log-seconds) required for each of the four methods, with computation performed on a 6-core 3.2 Ghz desktop computer. Since the BIP optimization in \eqref{eq:Tan} searches for the \textit{smallest design} for a fixed minimax criterion, instead of the \textit{smallest criterion} for a fixed design size, the timing for each BIP design is instead reported as the average time needed to generate all $n = 20, 30, \cdots, 100$-point designs. From Figure \ref{fig:time}, the computation time for \textsc{mMc-PSO} appears to be quite reasonable. For $p=2$, this time ranges from 15 to 90 seconds, whereas for $p=8$, this time ranges from 4 to 8 minutes. Not surprisingly, BIP takes the longest computation time, requiring nearly 30 minutes for each design. FFF designs can be computed faster than \textsc{mMc-PSO}, but provide inferior minimax performance since fewer approximating points can be used. Lastly, although principal points provide higher minimax distances than \textsc{mMc-PSO}, they can be computed the quickest of the four methods. These points can therefore be used as crude minimax designs when computation time is limited.

\subsection{Minimax designs on convex and bounded sets}

Next, we investigate the minimax performance of \textsc{mMc-PSO} for other convex and bounded design regions. Although much of existing literature considers designs on $[0,1]^p$, designs on other design regions are also of practical importance. For example, in studying the effects of temperature and pressure on injection molding, a hypercube design may be inappropriate since, from an engineering perspective, regions with high temperature and pressure may cause combustion of molding material, and experimental runs allocated in these regions therefore become wasted. \textsc{mMc-PSO} can be easily modified to generate minimax designs on design regions $\mathcal{X}$ which are convex and bounded. Convexity of $\mathcal{X}$ is necessary, since it ensures the $C_q$-centers updates in \textsc{mMc-PSO} remain in $\mathcal{X}$.

As mentioned previously, the key reason for using low-discrepancy sequences as the representative sample $\{\bm{y}_j\}_{j=1}^N$ is because such sequences provide a better approximation of the integral in \eqref{eq:minimaxk}. The question is how to generate these sequences for non-hypercube design regions, and to this end, this section is divided into two parts. First, when the Rosenblatt inverse transform for $U(\mathcal{X})$ (defined later) is easy to compute, there is an easy way to generate such sequences on $\mathcal{X}$. We illustrate this by computing minimax designs on the unit simplex and ball. When this transform is difficult to compute, uniform random sampling can be used as a last resort. This latter scenario is demonstrated using the motivating air quality example in Section 2.2.

\subsubsection{Minimax clustering using the Rosenblatt transform}

We begin by first defining the {Rosenblatt transform} $t_{\mathcal{X}}$:
\begin{definition}
Let $\mathcal{X} \subseteq \mathbb{R}^p$, and define the random vector $\bm{X} = (X_1, \cdots, X_p) \sim U(\mathcal{X})$. The Rosenblatt transform is defined as the transform $t_{\mathcal{X}}: \mathbb{R}^p \rightarrow \mathbb{R}^p$ satisfying:
\be
(x_1, \cdots, x_p) \mapsto (y_1, \cdots, y_p), \; \text{where} \; y_1 = F_1(x_1), \; y_i = F_i(x_i | x_1, \cdots, x_{i-1}), \quad i = 2, \cdots, p,
\label{eq:rosenblatt}
\ee
where $F_1(\cdot)$ is the distribution function (d.f.) of $X_1$, and $F_i(\cdot|x_1, \cdots, x_{i-1})$ is the conditional d.f. of $X_i$ given $X_1, \cdots, X_{i-1}$.
\end{definition}
It can be shown \citep{FW1993} that the inverse Rosenblatt transform of a low-discrepancy sequence on $[0,1]^p$ also has low-discrepancy on $\mathcal{X}$. Hence, when $t_{\mathcal{X}}^{-1}$ can be easily computed, minimax designs can be generated with Algorithm \ref{alg:mMcPSO} by simply taking the representative points $\{\bm{y}_j\}_{j=1}^N$ as the inverse transform of a Sobol' sequence.

Fortunately, when $\mathcal{X}$ is regularly-shaped, closed-form equations exist for the inverse Rosenblatt transform $t_{\mathcal{X}}^{-1}$. Transforms for common geometric shapes can be found in \cite{FW1993}. Using these equations, we generate minimax designs for the two regions:
\ben
\setlength\itemsep{0em}
\item The \textit{unit simplex} in $\mathbb{R}^p$: $A_p \equiv \{(x_1, \cdots, x_p) \in \mathbb{R}^p \; : \; 0 \leq x_1 \leq \cdots \leq x_p \leq 1\},$
\item The \textit{unit ball} in $\mathbb{R}^p$: $B_p \equiv \{(x_1, \cdots, x_p) \in \mathbb{R}^p \; : \; x_1^2 + \cdots + x_p^2 \leq 1 \}.$
\een

\begin{figure}[t]
\centering
\includegraphics[scale=0.55]{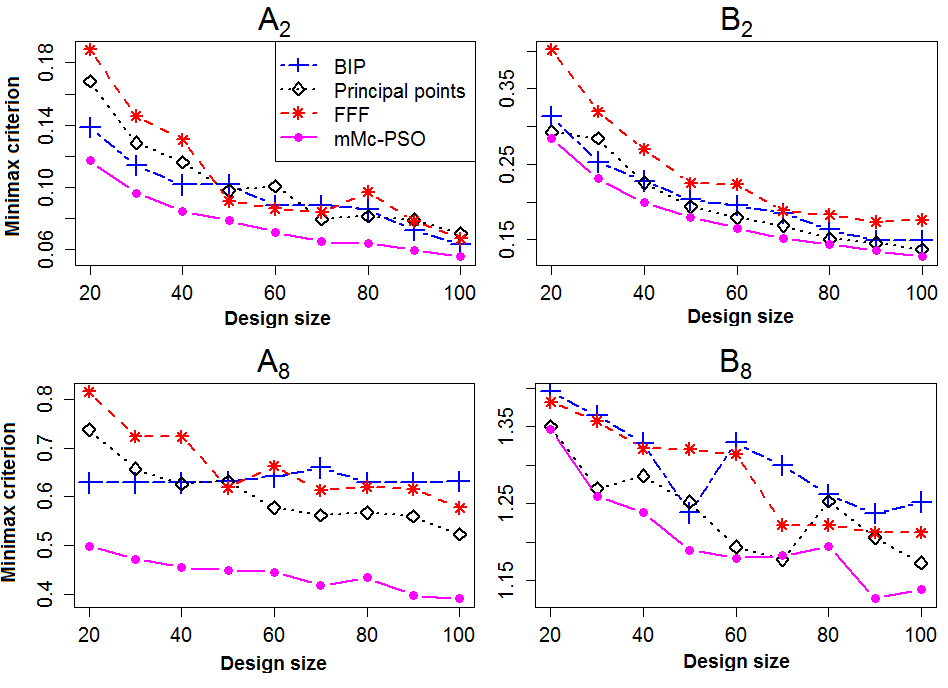}
\caption{Minimax criterion for various design sizes on $A_2$, $B_2$, $A_8$ and $B_8$. Designs from \textsc{mMc-PSO} consistently give the lowest minimax distance for nearly all design sizes.}
\label{fig:convcomp}
\end{figure}

\begin{figure}[t]
\includegraphics[scale=0.58]{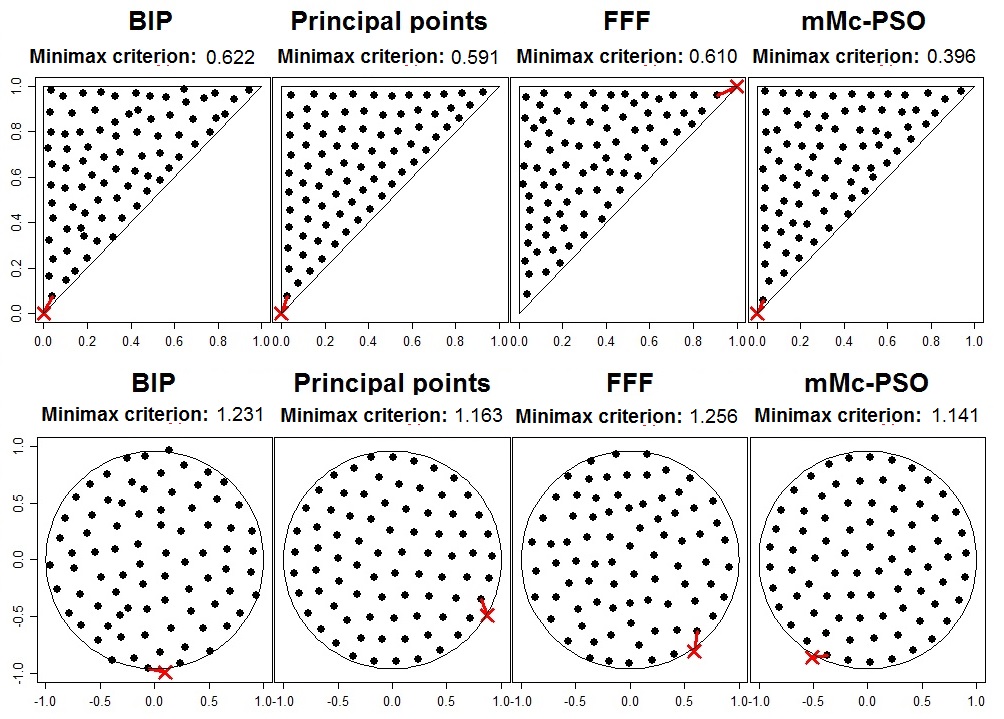}
\caption{Four different 80-point designs for $A_2$ and $B_2$. The red line connects the point in $\mathcal{X}$ furthest from the design (marked by `x') to its nearest design point, with its length equal to the minimax criterion. The proposed method \textsc{mMc-PSO} again provides the best minimax designs.}
\label{fig:80conv}
\end{figure}

The simulation settings are the same as before, with the exception that the candidate set for BIP is taken as the inverse transform of the first 1,000 points of a Sobol' sequence. Figure \ref{fig:convcomp} plots the minimax criterion of designs for $p=2$ and $p=8$, and Figure \ref{fig:80conv} plots the corresponding 80-point designs. Two interesting observations can be made. First, for both $p=2$ and $p=8$, \textsc{mMc-PSO} provides the best minimax designs for every design size $n$, which confirms the superiority of the proposed method in both low and high dimensions. Second, compared to principal points, \textsc{mMc-PSO} performs much better for the unit simplex $A_p$ compared to the unit ball $B_p$. This can be intuitively justified by the fact that both the arithmetic mean and $C_\infty$-center of a unit ball correspond to the same point, the center of the ball. However, when the design region is highly asymmetric, these two centers can indeed be quite different, which explains the sizable improvement of \textsc{mMc-PSO} over principal points for the unit simplex $A_p$.

\subsubsection{Back to the motivating example}

When $\mathcal{X}$ is irregularly-shaped, the inverse transform $t_{\mathcal{X}}^{-1}$ can be difficult to compute. In this case, the approximating points $\{\bm{y}_j\}_{j=1}^N$ can be generated using uniform random sampling on $\mathcal{X}$. We illustrate this using the earlier example of air quality monitoring in the state of Georgia. Note that, while the state of Georgia is not convex, it is ``convex enough'' to ensure $C_q$-centers remain in $\mathcal{X}$, so the proposed method can still be applied. 

\begin{figure}[t]
\centering
\begin{minipage}{0.47\textwidth}
\centering
\includegraphics[scale=0.55]{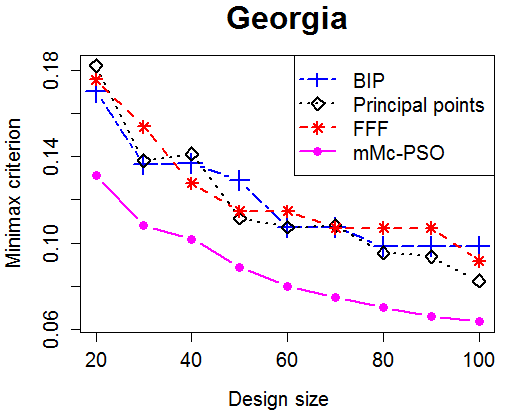}
\captionof{figure}{Minimax criterion for various design sizes on Georgia. Designs from \textsc{mMc-PSO} give the best minimax designs for all design sizes.}
\label{fig:GAcomp}
\end{minipage}
\hspace{0.1cm}
\begin{minipage}{0.47\textwidth}
\centering
\includegraphics[scale=0.60]{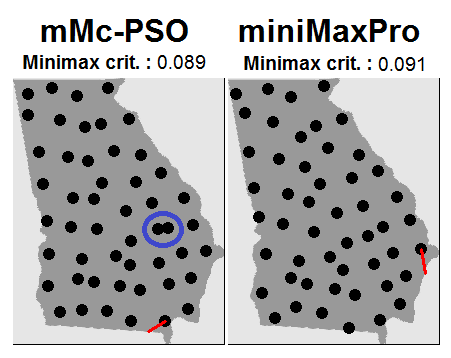}
\captionof{figure}{50-point designs on Georgia using \textsc{mMc-PSO} and \textsc{miniMaxPro}. The refinement step in the latter corrects some visual non-uniformities in the former design (circled in blue).}
\label{fig:pp}
\end{minipage}
\end{figure}

Figure \ref{fig:GAcomp} compares the minimax performance of $n = 20, 30, \cdots, 100$-point designs generated on Georgia, with the 20-point designs plotted in Figure \ref{fig:pp}. The simulation settings used here are the same as before. From the first figure, the minimax performance of \textsc{mMc-PSO} is sizably lower than existing methods for all design sizes, which illustrates the effectiveness of the proposed algorithm. One caveat of \textsc{mMc-PSO}, however, is that the generated designs appear visually non-uniform. For example, the 50-point design from \textsc{mMc-PSO} in the left plot of Figure \ref{fig:pp} shows several design points huddled closely together (such as the pair of points circled in blue), despite the design having a low minimax distance. One way to improve visual uniformity is to improve the uniformity of the design when projected onto the horizontal or vertical axis. This can be accomplished by performing the refinement step introduced in the following section. The right design in Figure \ref{fig:pp}, obtained by applying this refinement to the left design, is more visually uniform compared to the original design, despite having a slightly larger minimax distance. Users should therefore apply this refinement depending on whether visual uniformity or minimaxity is desired.

\section{Minimax projection designs}
As mentioned previously, minimax designs minimize the worst-case prediction error in computer experiment emulation \citep{Jea1990}. However, when a computer experiment has a large number of input variables, minimax designs as defined in \eqref{eq:minimax} may not be appropriate. This is because, by the \textit{effect sparsity} principle \citep{WH2011}, only a few of these inputs are expected to be active. Emulator designs in high dimensions should therefore provide not only good minimax performance on the full space $\mathcal{X}$, but also for \textit{projected subspaces} of $\mathcal{X}$. Recent developments in this vein include the MaxPro designs proposed by \cite{Jea2015}, which minimize the criterion:
\be
\sum_{i=1}^{n-1} \sum_{j=i+1}^n \frac{1}{d_{prod}(\bm{m}_i,\bm{m}_j)}, \quad d_{prod}(\bm{m}_i,\bm{m}_j) = \prod_{k=1}^p (m_{ik}-m_{jk})^2,
\label{eq:mp}
\ee
where $\bm{m}_i = (m_{i1}, \cdots, m_{ip})$ denotes the $i$-th design point. Extending this idea, we present below a new type of design called {minimax projection designs}, which are obtained by refining the minimax design from \textsc{mMc-PSO} using the MaxPro criterion in \eqref{eq:mp}.

\begin{minipage}{0.50\textwidth}
\begin{algorithm}[H]
\caption{Minimax projection designs}
\label{alg:mp}
\begin{algorithmic}[1]
\small
\Function{miniMaxPro}{$\cdots$}\\
\Comment{$\cdots$ - \textsc{mMc-PSO} params.}
\stb Generate an $n$-point minimax design $\mathcal{D} = \{\bm{m}_i\}_{i=1}^n \leftarrow$ \textsc{mMc-PSO}($\cdots$).
\Repeat
\For{$i = 1, \cdots, n$}
\stb Update $\{d_i\}_{i=1}^n$ in \eqref{eq:di}.
\stb Update $d^* = \max_{i} d_i$. 
\stb Update $\bm{m}_i$ by \eqref{eq:blockmp}.
\EndFor
\Until{design points converge.}
\stb \Return miniMaxPro design $\{\bm{m}_i\}_{i=1}^n$.
\EndFunction
\end{algorithmic}
\end{algorithm}
\end{minipage}
\hspace{0.1cm}
\begin{minipage}{0.45\textwidth}
\centering
\includegraphics[scale=0.62]{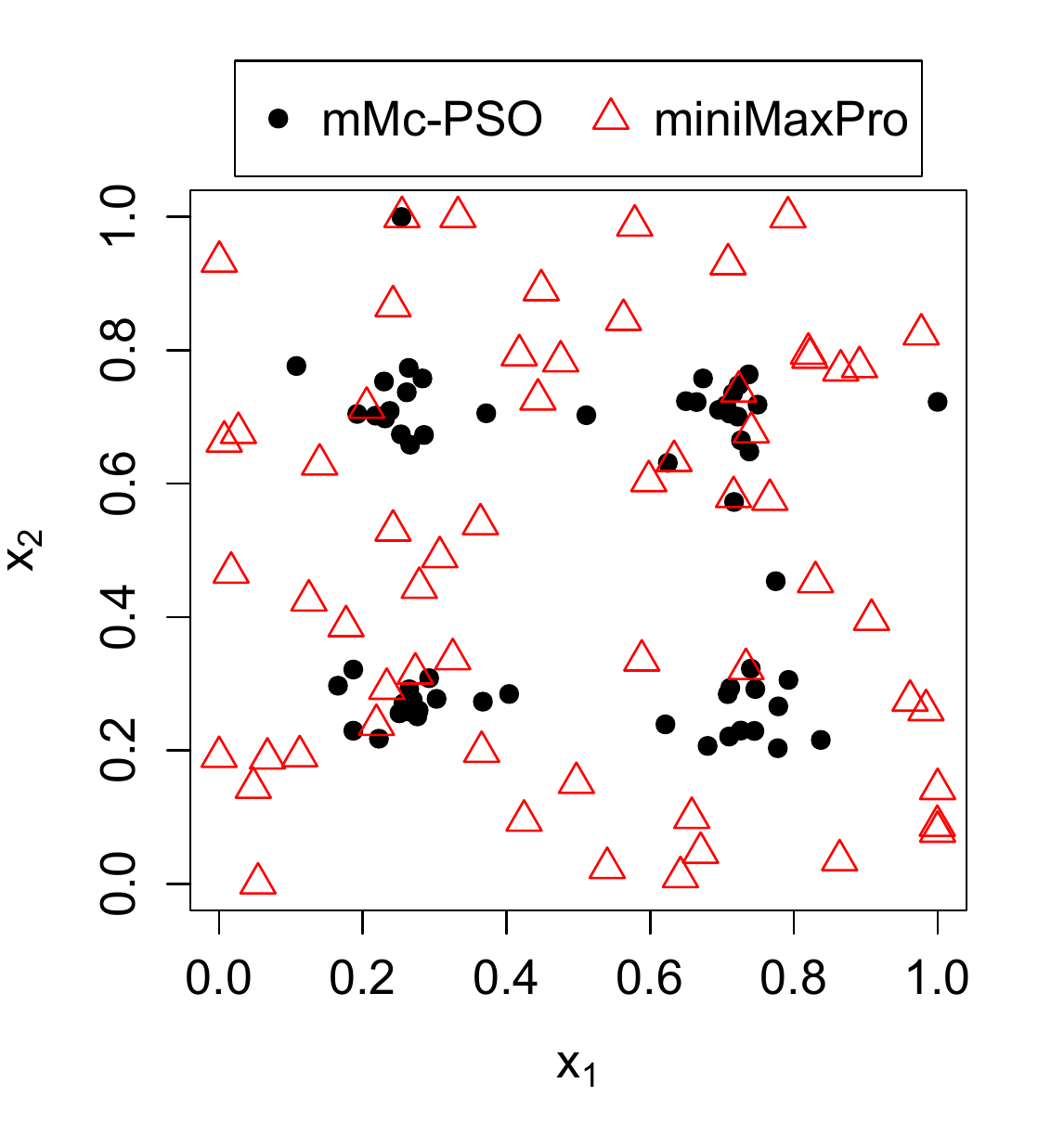}
\vspace{-1.5cm}
\captionof{figure}{A 2-d projection of 60-point \textsc{mMc-PSO} and {miniMaxPro} designs. The refinement step in \textsc{miniMaxPro} improves projected minimaxity.}
\label{fig:projdes}
\end{minipage}
\vspace{0.3cm}

In words, this refinement step \textit{improves} projected minimaxity while \textit{maintaining} the low minimax distance of the original \textsc{mMc-PSO} design. The details are as follows. Let $\mathcal{D}=\{\bm{m}_i\}_{i=1}^n$ be the design generated by \textsc{mMc-PSO}. Define the minimax distance of \textit{each} design point $\bm{m}_i$ as:
\be
d_i = \sup_{\bm{x} \in \mathcal{X}_i} \|\bm{x}-\bm{m}_i\|, \quad \text{where } \mathcal{X}_i = \{ \bm{x} \in \mathcal{X} : \|\bm{x}-\bm{m}_i\| \leq  \|\bm{x}-\bm{m}_j\|, \; \forall j = 1, \cdots, n\}
\label{eq:di}
\ee
is the collection of points in $\mathcal{X}$ closest in distance in $\bm{m}_i$. Note that the overall minimax distance in \eqref{eq:minimax} is simply the maximum of these distances, $d^* = \max_{i=1, \cdots, n} d_i$. For each point $\bm{m}_i$, the refinement step consist of two parts. First, compute the minimax distances $\{d_i\}_{i=1}^n$ and $d^*$. Next, update $\bm{m}_i$ by the optimization:
\be
\bm{m}_i \leftarrow \argmin_{\bm{m} \in \mathbb{R}^p}\sum_{j=1, j \neq i}^n \frac{1}{d_{prod}(\bm{m},\bm{m}_i)} \quad \text{s.t.} \quad \|\bm{m} - \bm{m}_i\| \leq d^* - d_i, \; \bm{m}_i \in \mathcal{X}.
\label{eq:blockmp}
\ee
This update can be viewed as the block-wise minimization of the MaxPro criterion \eqref{eq:mp} for the $i$-th design point $\bm{m}_i$, with the constraint $\|\bm{m} - \bm{m}_i\| \leq d^* - d_i$ ensuring the updated point is sufficiently close to the previous point. In our implementation, \eqref{eq:blockmp} is computed using the R package \texttt{nloptr} \citep{Ypm2014}. Repeating this two-stage refinement for each design point until convergence gives a point set which enjoys good space-filling properties after projections. Algorithm \ref{alg:mp} summarizes the detailed steps for generating this so-called minimax projection (miniMaxPro) design.

An appealing feature of miniMaxPro designs is that its projective space-fillingness does not come at a cost of increased minimax distance! That is, the minimax distance of the converged miniMaxPro design has the \textit{same} minimax distance on $\mathcal{X}$ as the original design from \textsc{mMc-PSO}. This is stated formally in the following proposition:
\begin{proposition}
When $\{d_i\}_{i=1}^n$ and $d^*$ are computed exactly, the two-stage refinement in lines 6 - 8 of Algorithm \ref{alg:mp} does not increase the minimax distance of $\mathcal{D}$ in line 2.
\label{prop:inv}
\end{proposition}
\noindent The proof of this proposition relies on the constraint $\|\bm{m} - \bm{m}_i\| \leq d^* - d_i$ in \eqref{eq:blockmp}; see Appendix for details. In practice, $\{d_i\}_{i=1}^n$ and $d^*$ are \textit{estimated} by appsroximating $\mathcal{X}$ using a finite representative set $\{\bm{y}_m\}_{m=1}^N$ (a Sobol' sequence is used in our implementation), so the overall minimax distance may increase after refinement. However, this increase is quite small when the number of approximating points $N$ is large (i.e., $N=10^5$), as shown in the simulations below.

To illustrate the effectiveness of this refinement, Figure \ref{fig:projdes} plots a two-dimensional projection of the 60-point design from \textsc{mMc-PSO} on $[0,1]^8$ and its corresponding miniMaxPro design. The \textsc{mMc-PSO} design clearly has poor minimax coverage after projection onto this 2-d subspace, with points closely focused around the four points $(0.5\pm 0.25,0.5 \pm 0.25)$. The miniMaxPro design, on the other hand, exhibits much better minimax performance after projection, which shows the refinement performs as intended.

Since one use of miniMaxPro designs is for computer experiment emulation, we compare its performance with two existing computer experiment designs: the MaxPro design \citep{Jea2015} and the FFF design \citep{Lea2014}. Three metrics are used to evaluate projective space-fillingness: $\text{mM}_k$, $\text{avg}_k$ and $\text{Mm}_k$, which are defined as:
\begin{align*}
\centering
\text{mM}_k &= \max_{r=1, \cdots, {p \choose k}} \sup_{\bm{x} \in \mathcal{P}_r(\mathcal{X})} \left\{ \frac{1}{n} \sum_{i=1}^n \frac{1}{\| \bm{x} - \mathcal{P}_r\bm{m}_i \|^{2k}}\right\}^{-1/(2k)},\\
\text{avg}_k &= \max_{r=1, \cdots, {p \choose k}} \int_{\mathcal{P}_r(\mathcal{X})} \| \bm{x} - Q(\bm{x},\{\mathcal{P}_r\bm{m}_i\}_{i=1}^n) \| \; d\bm{x} \text{ and }\\
\text{Mm}_k &= \min_{r=1, \cdots, {p \choose k}} \frac{1}{{n \choose 2}} \left\{ \sum_{i=1}^{n-1} \sum_{j=i+1}^n \frac{1}{\|\mathcal{P}_r\bm{m}_i - \mathcal{P}_r\bm{m}_j \|^{2k}} \right\}^{-1/(2k)}.
\end{align*}
Here, $r = 1, \cdots, {p \choose k}$ enumerates all projections of $\mathcal{X} \subseteq \mathbb{R}^p$ onto a subspace of dimension $k$, with $\mathcal{P}_r$ its corresponding projection operator. The metrics $\text{mM}_k$ and $\text{Mm}_k$ were proposed in \cite{Jea2015} to incorporate the minimax and maximin index of the design when projected into $k$ dimensions. The last metric $\text{avg}_k$ measures the average distance to a design point when projected into $k$ dimensions. Larger values of $\text{Mm}_k$ suggest better space-fillingness in terms of maximin, whereas smaller values of $\text{mM}_k$ and $\text{avg}_k$ indicate better space-fillingness in terms of minimax and average distance, respectively.

\begin{figure}[t]
\centering
\includegraphics[width=\textwidth]{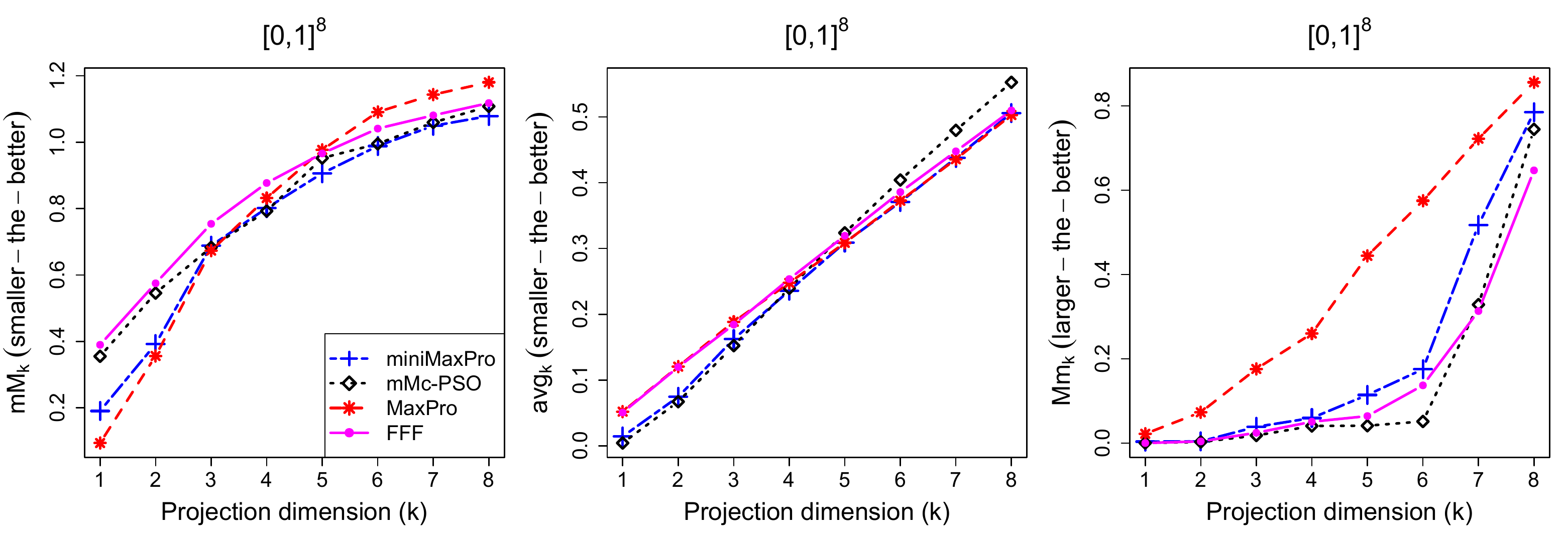}
\caption{$\text{mM}_k$, $\text{avg}_k$ and $\text{Mm}_k$ for four different 60-point designs on $[0,1]^8$. The proposed miniMaxPro design provides the best performance for $\text{mM}_k$ and $\text{avg}_k$, but performs worse for $\text{Mm}_k$.}
\label{fig:projcrit}
\end{figure}

Figure \ref{fig:projcrit} plots $\text{mM}_k$, $\text{avg}_k$ and $\text{Mm}_k$ for the 60-point MaxPro, FFF, miniMaxPro and the design from \textsc{mMc-PSO} (we refer to the latter as simply ``minimax design'' below). Similar results hold for other design sizes, and are not reported for brevity. For the minimax metric $\text{mM}_k$, both the miniMaxPro and minimax designs enjoy sizably improved performance in moderate dimensions ($4 \leq k \leq 8$). In lower dimensions ($1 \leq k \leq 3$), the refinement step for the miniMaxPro design allows it to be comparable with MaxPro. For the average distance metric $\text{avg}_k$, the miniMaxPro design appears to be the best choice over all projection dimensions. For the maximin metric $\text{Mm}_k$, the minimax and miniMaxPro designs give poorer performance to MaxPro. The refinement step for the latter, however, allows for sizable improvements with respect to maximin. To summarize, miniMaxPro designs appear to enjoy an improvement over existing designs in terms of projected minimax and average distance, but this comes at a cost of poorer performance for the projected maximin criterion.

\section{Discussion}

Minimax designs, by \textit{minimizing} the \textit{maximum} distance from any point in the design space $\mathcal{X} \subseteq \mathbb{R}^p$ to its closest design point, provide uniform coverage of $\mathcal{X}$ in the worst-case. Despite its many uses in computer experiments, optimal sensor placement and resource allocation problems, there have been little work on generating these designs efficiently. In this paper, we propose a new algorithm called \textsc{mMc-PSO} for computing minimax designs on convex and bounded design spaces, and demonstrate the efficiency of this method in low and highdimensions. Simulations on the unit hypercube, the unit simplex and ball, and the state of Georgia show that \textsc{mMc-PSO} provides better minimax designs compared to existing methods in literature. A new experimental design, called miniMaxPro designs, can then be constructed by refining the minimax design from \textsc{mMc-PSO} to ensure good projective space-fillingness. 

Despite the developments in this paper, there are still many avenues for further work. One of these is exploring the properties of minimax designs when the Euclidean norm is replaced by another norm for $\| \cdot \|$ in \eqref{eq:minimax}. Pursuing this may reveal better ways for generating designs in high-dimensions with good projective space-filling properties. Another direction is to explore more sophisticated hybridization schemes (e.g., \citealp{KL2002,Zea2009}) for incorporating PSO within clustering algorithms. This allows better minimax designs to be generated using less computational resources. 

\section*{Acknowledgement}
A C++ implementation of \textsc{mMc-PSO} and \textsc{miniMaxPro} is available in the R package \texttt{minimaxdesign} \citep{Mak2016}. This research is supported by the U. S. Army Research Office under grant number W911NF-14-1-0024.

\bibliography{references} 	

\begin{thebibliography}{}

\bibitem[\protect\citeauthoryear{Arthur, Manthey, and Roglin}{Arthur
  et~al.}{2009}]{Aea2009}
Arthur, D., B.~Manthey, and H.~Roglin (2009).
\newblock k-means has polynomial smoothed complexity.
\newblock pp.\  405--414.

\bibitem[\protect\citeauthoryear{Arthur and Vassilvitskii}{Arthur and
  Vassilvitskii}{2007}]{AS2007}
Arthur, D. and S.~Vassilvitskii (2007).
\newblock k-means++: The advantages of careful seeding.
\newblock pp.\  1027--1035.

\bibitem[\protect\citeauthoryear{Bell, Rosenfeld, Kim, Yoo, Lee, and
  Hahnenberger}{Bell et~al.}{2008}]{Bea2008}
Bell, T.~L., D.~Rosenfeld, K.-M. Kim, J.-M. Yoo, M.-I. Lee, and M.~Hahnenberger
  (2008).
\newblock Midweek increase in us summer rain and storm heights suggests air
  pollution invigorates rainstorms.
\newblock {\em Journal of Geophysical Research: Atmospheres\/}~{\em 113\/}(D2).

\bibitem[\protect\citeauthoryear{Ben-Tal and Nemirovski}{Ben-Tal and
  Nemirovski}{2001}]{BN2001}
Ben-Tal, A. and A.~Nemirovski (2001).
\newblock {\em Lectures on modern convex optimization: analysis, algorithms,
  and engineering applications}, Volume~2.
\newblock Siam.

\bibitem[\protect\citeauthoryear{Bertsekas}{Bertsekas}{1999}]{Ber1999}
Bertsekas, D.~P. (1999).
\newblock {\em Nonlinear programming}.
\newblock Athena scientific.

\bibitem[\protect\citeauthoryear{Bhowmick}{Bhowmick}{2009}]{Bho2009}
Bhowmick, A. (2009).
\newblock {\em A theoretical analysis of Lloyd’s algorithm for k-means
  clustering}.
\newblock Ph.\ D. thesis, Department of Computer Science and Engineering,
  Indian Institute of Technology, Kanpur.

\bibitem[\protect\citeauthoryear{Boyd and Vandenberghe}{Boyd and
  Vandenberghe}{2004}]{BV2004}
Boyd, S. and L.~Vandenberghe (2004).
\newblock {\em Convex optimization}.
\newblock Cambridge university press.

\bibitem[\protect\citeauthoryear{Christophe and Petr}{Christophe and
  Petr}{2014}]{DS2013}
Christophe, D. and S.~Petr (2014).
\newblock {\em randtoolbox: Generating and Testing Random Numbers}.
\newblock R package version 1.16.

\bibitem[\protect\citeauthoryear{Cox}{Cox}{1957}]{Cox1957}
Cox, D.~R. (1957).
\newblock Note on grouping.
\newblock {\em Journal of the American Statistical Association\/}~{\em
  52\/}(280), 543--547.

\bibitem[\protect\citeauthoryear{Dalenius}{Dalenius}{1950}]{Dal1950}
Dalenius, T. (1950).
\newblock The problem of optimum stratification.
\newblock {\em Scandinavian Actuarial Journal\/}~{\em 1950\/}(3-4), 203--213.

\bibitem[\protect\citeauthoryear{Eberhart and Kennedy}{Eberhart and
  Kennedy}{1995}]{EK1995}
Eberhart, R.~C. and J.~Kennedy (1995).
\newblock A new optimizer using particle swarm theory.
\newblock ~{\em 1}, 39--43.

\bibitem[\protect\citeauthoryear{Eppstein}{Eppstein}{2000}]{Epp2000}
Eppstein, D. (2000).
\newblock Fast hierarchical clustering and other applications of dynamic
  closest pairs.
\newblock {\em Journal of Experimental Algorithmics (JEA)\/}~{\em 5}, 1.

\bibitem[\protect\citeauthoryear{Fang and Wang}{Fang and Wang}{1993}]{FW1993}
Fang, K.-T. and Y.~Wang (1993).
\newblock {\em Number-theoretic methods in statistics}, Volume~51.
\newblock CRC Press.

\bibitem[\protect\citeauthoryear{Flury}{Flury}{1990}]{Flu1990}
Flury, B.~A. (1990).
\newblock Principal points.
\newblock {\em Biometrika\/}~{\em 77\/}(1), 33--41.

\bibitem[\protect\citeauthoryear{Flury}{Flury}{1993}]{Flu1993}
Flury, B.~D. (1993).
\newblock Estimation of principal points.
\newblock {\em Applied Statistics\/}, 139--151.

\bibitem[\protect\citeauthoryear{Jain}{Jain}{2010}]{Jai2010}
Jain, A.~K. (2010).
\newblock Data clustering: 50 years beyond k-means.
\newblock {\em Pattern recognition letters\/}~{\em 31\/}(8), 651--666.

\bibitem[\protect\citeauthoryear{John, Johnson, Moore, and Ylvisaker}{John
  et~al.}{1995}]{Jea1995}
John, P., M.~Johnson, L.~Moore, and D.~Ylvisaker (1995).
\newblock Minimax distance designs in two-level factorial experiments.
\newblock {\em Journal of Statistical Planning and Inference\/}~{\em 44\/}(2),
  249--263.

\bibitem[\protect\citeauthoryear{Johnson, Moore, and Ylvisaker}{Johnson
  et~al.}{1990}]{Jea1990}
Johnson, M.~E., L.~M. Moore, and D.~Ylvisaker (1990).
\newblock Minimax and maximin distance designs.
\newblock {\em Journal of statistical planning and inference\/}~{\em 26\/}(2),
  131--148.

\bibitem[\protect\citeauthoryear{Joseph, Gul, and Ba}{Joseph
  et~al.}{2015}]{Jea2015}
Joseph, V.~R., E.~Gul, and S.~Ba (2015).
\newblock Maximum projection designs for computer experiments.
\newblock {\em Biometrika\/}~{\em 102}, 371--380.

\bibitem[\protect\citeauthoryear{Krink and L{\o}vbjerg}{Krink and
  L{\o}vbjerg}{2002}]{KL2002}
Krink, T. and M.~L{\o}vbjerg (2002).
\newblock The lifecycle model: combining particle swarm optimisation, genetic
  algorithms and hillclimbers.
\newblock In {\em Parallel Problem Solving from Nature—PPSN VII}, pp.\
  621--630. Springer.

\bibitem[\protect\citeauthoryear{Lekivetz and Jones}{Lekivetz and
  Jones}{2015}]{Lea2014}
Lekivetz, R. and B.~Jones (2015).
\newblock Fast flexible space-filling designs for nonrectangular regions.
\newblock {\em Quality and Reliability Engineering International\/}~{\em
  31\/}(5), 829--837.

\bibitem[\protect\citeauthoryear{Likas, Vlassis, and Verbeek}{Likas
  et~al.}{2003}]{LVV2003}
Likas, A., N.~Vlassis, and J.~J. Verbeek (2003).
\newblock The global k-means clustering algorithm.
\newblock {\em Pattern recognition\/}~{\em 36\/}(2), 451--461.

\bibitem[\protect\citeauthoryear{Linde, Buzo, and Gray}{Linde
  et~al.}{1980}]{LBG1980}
Linde, Y., A.~Buzo, and R.~M. Gray (1980).
\newblock An algorithm for vector quantizer design.
\newblock {\em Communications, IEEE Transactions on\/}~{\em 28\/}(1), 84--95.

\bibitem[\protect\citeauthoryear{Lloyd}{Lloyd}{1957}]{Llo1957}
Lloyd, S. (1957).
\newblock Binary block coding.
\newblock {\em Bell System Technical Journal\/}~{\em 36\/}(2), 517--535.

\bibitem[\protect\citeauthoryear{Lloyd}{Lloyd}{1982}]{Llo1982}
Lloyd, S. (1982).
\newblock Least squares quantization in pcm.
\newblock {\em Information Theory, IEEE Transactions on\/}~{\em 28\/}(2),
  129--137.

\bibitem[\protect\citeauthoryear{Luss}{Luss}{1999}]{Lus1999}
Luss, H. (1999).
\newblock On equitable resource allocation problems: a lexicographic minimax
  approach.
\newblock {\em Operations Research\/}~{\em 47\/}(3), 361--378.

\bibitem[\protect\citeauthoryear{Mak}{Mak}{2016}]{Mak2016}
Mak, S. (2016).
\newblock {\em minimaxdesign: Minimax and Minimax Projection Designs}.
\newblock R package version 0.1.0.

\bibitem[\protect\citeauthoryear{Nesterov}{Nesterov}{1983}]{Nes1983}
Nesterov, Y. (1983).
\newblock A method of solving a convex programming problem with convergence
  rate o (1/k2).
\newblock ~{\em 27\/}(2), 372--376.

\bibitem[\protect\citeauthoryear{Nesterov}{Nesterov}{2007}]{Nes2007}
Nesterov, Y. (2007).
\newblock Gradient methods for minimizing composite objective function.
\newblock Technical report, UCL.

\bibitem[\protect\citeauthoryear{Nesterov}{Nesterov}{2013}]{Nes2013}
Nesterov, Y. (2013).
\newblock {\em Introductory lectures on convex optimization: A basic course},
  Volume~87.
\newblock Springer Science \& Business Media.

\bibitem[\protect\citeauthoryear{Niederreiter}{Niederreiter}{1992}]{Nie1992}
Niederreiter, H. (1992).
\newblock {\em Random number generation and quasi-Monte Carlo methods},
  Volume~63.
\newblock SIAM.

\bibitem[\protect\citeauthoryear{Nielsen and Bhatia}{Nielsen and
  Bhatia}{2013}]{NB2013}
Nielsen, F. and R.~Bhatia (2013).
\newblock {\em Matrix information geometry}.
\newblock Springer.

\bibitem[\protect\citeauthoryear{Nocedal and Wright}{Nocedal and
  Wright}{2006}]{NW2006}
Nocedal, J. and S.~Wright (2006).
\newblock {\em Numerical optimization}.
\newblock Springer Science \& Business Media.

\bibitem[\protect\citeauthoryear{Oser}{Oser}{2016}]{Ose2016}
Oser, D. (2016).
\newblock Georgia air monitoring.
\newblock \url{http://amp.georgiaair.org/}.
\newblock Accessed: 2016-10-15.

\bibitem[\protect\citeauthoryear{Owen}{Owen}{1995}]{Owe1995}
Owen, A.~B. (1995).
\newblock {\em Randomly permuted (t, m, s)-nets and (t, s)-sequences}.
\newblock Springer.

\bibitem[\protect\citeauthoryear{Patan}{Patan}{2012}]{Pat2012}
Patan, M. (2012).
\newblock {\em Optimal sensor networks scheduling in identification of
  distributed parameter systems}, Volume 425.
\newblock Springer Science \& Business Media.

\bibitem[\protect\citeauthoryear{Sobol}{Sobol}{1967}]{Sob1967}
Sobol, I.~M. (1967).
\newblock On the distribution of points in a cube and the approximate
  evaluation of integrals.
\newblock {\em USSR Computational mathematics and mathematical physics\/}~(7),
  86--112.

\bibitem[\protect\citeauthoryear{Tan}{Tan}{2013}]{Tan2013}
Tan, M.~H. (2013).
\newblock Minimax designs for finite design regions.
\newblock {\em Technometrics\/}~{\em 55\/}(3), 346--358.

\bibitem[\protect\citeauthoryear{Tzortzis and Likas}{Tzortzis and
  Likas}{2009}]{TL2009}
Tzortzis, G.~F. and C.~Likas (2009).
\newblock The global kernel-means algorithm for clustering in feature space.
\newblock {\em Neural Networks, IEEE Transactions on\/}~{\em 20\/}(7),
  1181--1194.

\bibitem[\protect\citeauthoryear{van Dam}{van Dam}{2008}]{Van2008}
van Dam, E.~R. (2008).
\newblock Two-dimensional minimax latin hypercube designs.
\newblock {\em Discrete Applied Mathematics\/}~{\em 156\/}(18), 3483--3493.

\bibitem[\protect\citeauthoryear{Van~der Merwe and Engelbrecht}{Van~der Merwe
  and Engelbrecht}{2003}]{VE2003}
Van~der Merwe, D. and A.~P. Engelbrecht (2003).
\newblock Data clustering using particle swarm optimization.
\newblock ~{\em 1}, 215--220.

\bibitem[\protect\citeauthoryear{Vanli, Zhang, Nguyen, and Wang}{Vanli
  et~al.}{2012}]{Vea2012}
Vanli, O.~A., C.~Zhang, A.~Nguyen, and B.~Wang (2012).
\newblock A minimax sensor placement approach for damage detection in composite
  structures.
\newblock {\em Journal of Intelligent Material Systems and Structures\/},
  1045389X12440751.

\bibitem[\protect\citeauthoryear{Ward~Jr}{Ward~Jr}{1963}]{War1963}
Ward~Jr, J.~H. (1963).
\newblock Hierarchical grouping to optimize an objective function.
\newblock {\em Journal of the American statistical association\/}~{\em
  58\/}(301), 236--244.

\bibitem[\protect\citeauthoryear{Wu and Hamada}{Wu and Hamada}{2011}]{WH2011}
Wu, C. F.~J. and M.~S. Hamada (2011).
\newblock {\em Experiments: planning, analysis, and optimization}, Volume 552.
\newblock John Wiley \& Sons.

\bibitem[\protect\citeauthoryear{Ypma}{Ypma}{2014}]{Ypm2014}
Ypma, J. (2014).
\newblock nloptr: R interface to nlopt.
\newblock {\em R package version\/}~{\em 1\/}(4).

\bibitem[\protect\citeauthoryear{Zhan, Zhang, Li, and Chung}{Zhan
  et~al.}{2009}]{Zea2009}
Zhan, Z.-H., J.~Zhang, Y.~Li, and H.~S.-H. Chung (2009).
\newblock Adaptive particle swarm optimization.
\newblock {\em Systems, Man, and Cybernetics, Part B: Cybernetics, IEEE
  Transactions on\/}~{\em 39\/}(6), 1362--1381.

\end{thebibliography}

\pagebreak

\begin{appendices}
\numberwithin{equation}{section}
\counterwithin{figure}{section} 
\setcounter{page}{1}

\section{Proofs}

\subsection*{Proof of Theorem \ref{thm:uni}}

\begin{lemma}
Let $h: \mathbb{R}^p \rightarrow \mathbb{R}_+$ be a strictly convex function, and let $g: \mathbb{R}_+ \rightarrow \mathbb{R}_+$ be a convex and strictly increasing function. Then the composition $g \circ h: \mathbb{R}^p \rightarrow \mathbb{R}_+$ is strictly convex.
\label{lem:convex}
\end{lemma}

\begin{proof}(Lemma \ref{lem:convex})
This is easy to show using first principles. Let $\alpha \in (0,1)$ and let $\bm{z} \neq \bm{z}'$ be two points in $\mathbb{R}^p$. By strict convexity, we have:
\[h(\alpha \bm{z} + (1 - \alpha) \bm{z}') < \alpha h( \bm{z} ) + (1 - \alpha) h(\bm{z}').\]
Moreover, since $g$ is strictly increasing and convex, it follows that:
\[(g \circ h) (\alpha \bm{z} + (1 - \alpha) \bm{z}') < g( \alpha h( \bm{z} ) + (1 - \alpha) h(\bm{z}')) \leq \alpha (g \circ h) (\bm{z}) + (1 - \alpha) (g \circ h) (\bm{z}'),\]
which proves the strict convexity of $g \circ h$.
\end{proof}

\begin{proof}(Theorem \ref{thm:uni}) Let $g(x) = x^{q/2}$ and $h(\bm{z}) = \|\bm{z} - \bm{z}_i\|_2^2$. It is easy to verify that $h$ is strictly convex, and $g$ is convex and strictly increasing on $\mathbb{R}_+$. By Lemma 1, it follows that $(g \circ f)(\bm{x}) = \|\bm{z} - \bm{z}_i\|_2^q$ is strictly convex. Hence, for any $\alpha \in (0,1)$ and $\bm{z}, \bm{z}' \in \mathbb{R}^p, \bm{z} \neq \bm{z}'$, we have:
\begin{align*}
D_{q}(\alpha \bm{z} + (1 - \alpha) \bm{z}'; \mathcal{Z}) &= \frac{1}{mq}\sum_{i = 1}^n \| \left\{ \left(\alpha \bm{z} + (1 - \alpha) \bm{z}' \right) - \bm{z}_i \right\} \|_2^q\\
&< \frac{1}{mq} \sum_{i = 1}^n \left\{ \alpha \| \bm{z} - \bm{z}_i\|_2^q + (1 - \alpha) \| \bm{z}' - \bm{z}_i\|_2^q \right\}\\
&= \alpha D_{q}( \bm{z} ; \mathcal{Z}) + (1-\alpha) D_{q}( \bm{z}'; \mathcal{Z}),
\end{align*}
so the objective $D_{q}(\bm{z}; \mathcal{Z})$ is strictly convex in $\bm{z}$.

Using this fact, we show that \eqref{eq:ckcent} has a unique minimizer. Note that the objective $D_{q}(\bm{z}; \mathcal{Z})$ is continuous and coercive on the closed set $\mathbb{R}^p$, where the latter term implies that for all sequences $\{\bm{z}_k\}_{k=1}^\infty$ satisfying $\|\bm{z}_k\|_2 \rightarrow \infty$, $\lim_{k \rightarrow \infty} D_{q}(\bm{z}_k; \mathcal{Z}) = \infty.$ It follows from Proposition A.8 in \cite{Ber1999} and the strict convexity of $D_{q}(\bm{z}; \mathcal{Z})$ that there exists exactly one one global minimum of \eqref{eq:ckcent}, so $C_q(\mathcal{Z})$ is uniquely defined.

To prove that the unique minimizer $C_q(\mathcal{Z})$ is contained in $\text{conv}(\mathcal{Z})$, note that by first-order optimality conditions, $C_q(\mathcal{Z})$ must satisfy:
\begin{align*}
\nabla D_{q}(C_q(\mathcal{Z}); \mathcal{Z}) &= \frac{1}{n} \sum_{i=1}^m \left\{\|C_q(\mathcal{Z}) - \bm{z}_i\|_2^{q - 2} (C_q(\mathcal{Z}) - \bm{z}_i) \right\} = \bm{0}\\
&\Leftrightarrow C_q(\mathcal{Z}) = \sum_{i=1}^m \left\{ \frac{\|C_q(\mathcal{Z}) - \bm{z}_i\|_2^{q - 2}}{\sum_{j=1}^n \|C_q(\mathcal{Z}) - \bm{z}_j\|_2^{q - 2} } \; \bm{z}_i \right\} \equiv \sum_{i=1}^m \alpha_i \bm{z}_i.
\end{align*}
Since the weights $\{\alpha_i\}_{i=1}^m$ satisfy $\alpha_i \geq 0$ and $\sum_{i=1}^m \alpha_i = 1$, it follows by definition that $C_q(\mathcal{Z}) \in \text{conv}(\mathcal{Z})$, which is as desired.

\end{proof}

\subsection*{Proof of Theorem \ref{thm:lip}}

\begin{lemma}
Let $\mathcal{Z} = \{\bm{z}_i\}_{i=1}^m$ be a set of points in $\mathbb{R}^p$. Then there exists some point $\bm{z}_j \in \mathcal{Z}$ such that $D_q(\bm{z}_j;\mathcal{Z}) \geq  D_q(\bm{z};\mathcal{Z})$ for all $\bm{z} \in \text{\normalsize{conv}}(\mathcal{Z})$.
\label{lem:bas}
\end{lemma}

\begin{proof}(Lemma \ref{lem:bas})
Since conv$(\mathcal{Z})$ is a compact set, the set of maximizers in:
\[\mathcal{M} = \text{argmax}_{\bm{z} \in \text{conv}(\mathcal{Z})} D_q(\bm{z};\mathcal{Z})\]
is non-empty, so an equivalent claim is that $\bm{z}_j \in \mathcal{M}$ for some $j=1, \cdots, m$. Suppose, for contradiction, that $\bm{z}_j \notin \mathcal{M}$ for all $j = 1, \cdots, m$, and let $\bm{z}' = \sum_{i = 1}^m \alpha_j \bm{z}_j \notin \mathcal{Z}$ be a maximizer in $\mathcal{M}$, with $\alpha_j \geq 0$ and $\sum_{j=1}^m \alpha_j = 1$. Then, by convexity, we have:
\small
\begin{align*}
D_q(\bm{z}'; \mathcal{Z}) = \frac{1}{mq}\sum_{i=1}^m \left\|  \sum_{j = 1}^m \alpha_j (\bm{z}_j - \bm{z}_i) \right\|_2^q \leq \frac{1}{mq} \sum_{i=1}^m \sum_{j=1}^m \alpha_j \|\bm{z}_j - \bm{z}_i\|_2^q &= \frac{1}{mq} \sum_{j=1}^m \alpha_j \left( \sum_{i=1}^m \|\bm{z}_j - \bm{z}_i\|_2^q\right)\\
&= \sum_{j=1}^m \alpha_j D_q(\bm{z}_j;\mathcal{Z}),
\end{align*}
\normalsize
which implies that $D_q(\bm{z}'; \mathcal{Z}) \leq D_q(\bm{z}_j; \mathcal{Z})$ for at least one $j = 1, \cdots, m$. Since $\bm{z}' \in \mathcal{M}$, this implies that $\bm{z}_j \in \mathcal{M}$, which is a contradiction. The lemma therefore holds.
\end{proof}

\begin{proof}(Theorem \ref{thm:lip})
Since $D_{q}( \bm{z}; \mathcal{Z})$ is twice-differentiable, it is $\beta$-smooth on $\text{conv}(\mathcal{Z})$ if and only if:
\begin{equation}
\nabla^2 D_{q}( \bm{z}; \mathcal{Z}) \preceq \beta \bm{I} \quad \text{for all $\bm{z} \in \text{conv}(\mathcal{Z})$}.
\label{eq:lip}
\end{equation}
Letting $\lambda_{max}\{\bm{A}\}$ denote the largest eigenvalue of $\bm{A}$, it follows that:
\begin{align*}
\lambda_{max}\{\nabla^2 D_{q}( \bm{z}; \mathcal{Z})\} &= \lambda_{max}\left\{\frac{q-2}{m} \sum_{i=1}^m \left\{ \|\bm{z} - \bm{z}_i\|_2^{q-4} (\bm{z} - \bm{z}_i)(\bm{z} - \bm{z}_i)^T\right\} + \frac{1}{m} \sum_{i=1}^m \|\bm{z}-\bm{z}_i\|_2^{q-2} \bm{I}\right\}\\
& \leq \frac{q-2}{m} \sum_{i=1}^m \|\bm{z} - \bm{z}_i\|_2^{q-4} \lambda_{max}\left\{(\bm{z} - \bm{z}_i)(\bm{z} - \bm{z}_i)^T\right\}+ \frac{1}{m} \sum_{i=1}^m \|\bm{z}-\bm{z}_i\|_2^{q-2} \lambda_{max}\{\bm{I}\}\\
& = \frac{q-2}{m} \sum_{i=1}^m \|\bm{z} - \bm{z}_i\|_2^{q-4} \cdot \|\bm{z} - \bm{z}_i\|_2^{2} + \frac{1}{m} \sum_{i=1}^m \|\bm{z}-\bm{z}_i\|_2^{q-2}\\
& = \frac{q-1}{m} \sum_{i=1}^m \|\bm{z} - \bm{z}_i\|_2^{q-2} \leq \frac{q-1}{m} \max_{j=1, \cdots, m}\sum_{i=1}^m \|\bm{z}_j - \bm{z}_i\|_2^{q-2} = \bar{\beta},
\end{align*}
where the last inequality holds by Lemma \ref{lem:bas}. Hence, $\nabla^2 D_{q}( \bm{z}; \mathcal{Z}) \preceq \bar{\beta} \bm{I}$ for all $\bm{z} \in \text{conv}(\mathcal{Z})$, so $D_{q}( \bm{z}; \mathcal{Z})$ is $\bar{\beta}$-smooth on $\text{conv}(\mathcal{Z})$ by \eqref{eq:lip}.

Likewise, since $D_{q}( \bm{z}; \mathcal{Z})$ is twice-differentiable, it is $\mu$-strongly convex on $\text{conv}(\mathcal{Z})$ if and only if:
\begin{equation}
 \mu \bm{I} \preceq \nabla^2 D_{q}( \bm{z}; \mathcal{Z}) \quad \text{for all $\bm{z} \in \text{conv}(\mathcal{Z})$}.
\label{eq:sc}
\end{equation}
Letting $\lambda_{min}\{\bm{A}\}$ denote the smallest eigenvalue of $\bm{A}$, we have:
\small
\begin{align*}
\lambda_{min}\{\nabla^2 D_{q}( \bm{z}; \mathcal{Z})\} &= \lambda_{min}\left\{\frac{q-2}{m} \sum_{i=1}^m \left\{ \|\bm{z} - \bm{z}_i\|_2^{q-4} (\bm{z} - \bm{z}_i)(\bm{z} - \bm{z}_i)^T\right\} + \frac{1}{m} \sum_{i=1}^m \|\bm{z}-\bm{z}_i\|_2^{q-2} \bm{I}\right\}\\
& \geq \frac{q-2}{m} \sum_{i=1}^m \|\bm{z} - \bm{z}_i\|_2^{q-4} \lambda_{min}\left\{(\bm{z} - \bm{z}_i)(\bm{z} - \bm{z}_i)^T\right\}+ \frac{1}{m} \sum_{i=1}^m \|\bm{z}-\bm{z}_i\|_2^{q-2} \lambda_{min}\{\bm{I}\}\\
& \geq \frac{q-2}{m} \sum_{i=1}^m \|\bm{z} - \bm{z}_i\|_2^{q-4} \cdot 0 + \frac{1}{m} \sum_{i=1}^m \|\bm{z}-\bm{z}_i\|_2^{q-2}\geq \frac{1}{m} \sum_{i=1}^m \|C_{q-2}(\mathcal{Z})-\bm{z}_i\|_2^{q-2} = \bar{\mu},
\end{align*}
\normalsize
where the last inequality holds by definition of $C_{q-2}(\mathcal{Z})$. Hence by \eqref{eq:sc}, $D_{q}( \bm{z}; \mathcal{Z})$ is $\bar{\mu}$-strongly convex.
\end{proof}

\subsection*{Proof of Corollary \ref{corr:cqconv}}

Consider a $\beta$-smooth and $\mu$-strongly convex function $h$ with unique minimizer $\bm{u}^*$. It can be shown \citep{Nes2007} that an iteration upper bound of $t = O\left( \sqrt{\frac{\beta}{\mu}} \log \frac{1}{\epsilon_{in}} \right)$ guarantees an $\epsilon_{in}$-accuracy in objective, i.e. $|h(\bm{u}^{[t]}) - h(\bm{u}^*)| < \epsilon_{in}$. Combining this iteration bound with the result in Theorem \ref{thm:lip}, and using the fact that each update requires $O(mp)$ work, we get the desired result.

\subsection*{Proof of Theorem \ref{thm:mMconv}}

The three parts of this theorem are individually easy to verify. For finite termination, we showed in Section 3.1 that the objective in \eqref{eq:clustk} strictly decreases after each loop iteration of Algorithm \ref{alg:mMc}. Moreover, there are exactly $N^n$ possible assignments of the sample $\{\bm{y}_j\}_{j=1}^N$ to the design points $\{\bm{m}_i\}_{i=1}^n$. Suppose, for contradiction, that Algorithm \ref{alg:mMc} does not terminate after $N^n$ iterations. Then there exists at least two iterations which begin with the same assignment of $\{\bm{y}_j\}_{j=1}^N$. This, in turn, generates the same design $\{\bm{m}_i\}_{i=1}^n$ at the end of both iterations, which presents a contradiction to the strictly decreasing objective values induced by each loop iteration of Algorithm \ref{alg:mMc}. The first claim therefore holds.

Next, regarding running time, consider the two updates in a single loop iteration of Algorithm \ref{alg:mMc}. The first update assigns each sample point in $\{\bm{y}_j\}$ to its closest design point, which requires $O(Nnp)$ work. The second update computes, for each design point, the $C_q$-center of samples assigned to it. Let $\mathcal{Z} = \{\bm{z}_j\}_{j=1}^{m_i}$ be the $m_i$ points assigned to the $i$-th design point. From Corollary \ref{corr:cqconv}, the computation of its $C_q$-center requires $O(m_i p \sqrt{(q-1) \kappa_{q-2}(\mathcal{Z}) \log (1/\epsilon_{in})})$ work. Letting $\tilde{\bm{z}} = \text{argmax}_{j=1, \cdots, m_i} D_{q}( \bm{z}_j; \mathcal{Z})$, it follows that for any $q \geq 2$:
\begin{align*}
\kappa_q(\mathcal{Z}) = \frac{D_q(\tilde{\bm{z}}; \mathcal{Z})}{D_{q}( C_{q}(\mathcal{Z}); \mathcal{Z})} &\leq \frac{\sum_{i=1}^{m_i} \|\bm{z}_i - C_{q}(\mathcal{Z})\|_2^q + m_i \|\tilde{\bm{z}} - C_{q}(\mathcal{Z})\|_2^q}{\sum_{i=1}^{m_i} \|\bm{z}_i - C_{q}(\mathcal{Z})\|_2^q}\\
& \leq 1+\frac{m_i \|\tilde{\bm{z}} - C_{q}(\mathcal{Z})\|_2^q}{\sum_{i=1}^{m_i} \|\bm{z}_i - C_{q}(\mathcal{Z})\|_2^q} \leq m_i+1.
\end{align*}
Hence, updating $C_q$-centers for all $n$ design points require a total work of:
\small
\begin{align*}
\sum_{i=1}^n O(m_i p \sqrt{(q-1) \kappa_{q-2}(\mathcal{Z}) \log (1/\epsilon_{in})}) &\leq O\left( \left\{\sum_{i=1}^n m_i^{3/2}\right\} p\sqrt{q-1} \log \frac{1}{\epsilon_{in}}\right)\\
& \leq O\left( \left\{\sum_{i=1}^n m_i \right\}^{3/2} p\sqrt{q-1} \log \frac{1}{\epsilon_{in}}\right)\\
&= O\left( N^{3/2} p\sqrt{q-1} \log \frac{1}{\epsilon_{in}}\right).
\end{align*}
\normalsize
Finally, since $n \leq N^{1/2}$, the running time of the second step dominates the first, which completes the argument.

Finally, assume that the $C_q$-center updates in \eqref{eq:ckcent} are exact. By the termination conditions of Algorithm \ref{alg:mMc}, the converged design is optimal given fixed assignments, and the converged assignment variables are optimal given a fixed design. Hence, the converged design (as well as its corresponding assignment) are locally optimal for \eqref{eq:clustk}.

\subsection*{Proof of Proposition \ref{prop:inv}}

This can be shown by a simple application of the triangle inequality. Let $\mathcal{D} = \{\bm{m}_i\}_{i=1}^n$ be the design at the current iteration, and without loss of generality, suppose the first design point $\bm{m}_1$ is to be updated. Also, let $\{d_i\}_{i=1}^n$ be the minimax distances for each design point (defined in \eqref{eq:di}), with $d^* = \max_i d_i$ being the overall minimax distance of $\mathcal{D}$.

Let $\tilde{\bm{m}}_1$ be the optimal design point in \eqref{eq:blockmp}, and note that, by optimization constraints, $\|\tilde{\bm{m}}_1 - \bm{m}_1\| \leq d^* - d_1$. Denoting $\tilde{d}^*$ as the overall minimax distance of the new design $\tilde{\mathcal{D}} = \{\tilde{\bm{m}}_1, \bm{m}_2, \cdots, \bm{m}_n\}$, the claim is that $\tilde{d}^* \leq d^*$. To prove this, let $\bm{x}$ be the point in $\mathcal{X}$ achieving the minimax distance $\tilde{d}^*$, and consider the following three cases:
\bi
\item If $Q(\bm{x},\tilde{\mathcal{D}})$, the closest design point to $\bm{x}$ in $\tilde{\mathcal{D}}$, equals $\tilde{\bm{m}}_1$, then:
\[\tilde{d}^*  = \|\bm{x} -  \tilde{\bm{m}_1}\| \leq \|\bm{x} - \bm{m}_1\| + \|\bm{m}_1-\tilde{\bm{m}}_1\| \leq d_1 + (d^* - d_1) = d^*.\]
\item If $Q(\bm{x},\tilde{\mathcal{D}}) = \bm{m}_i$ for some $i = 2, \cdots, n$, and $Q(\bm{x},{\mathcal{D}}) = \bm{m}_1$, then:
\[\tilde{d}^*  = \|\bm{x} -  {\bm{m}_i}\| \leq \|\bm{x} - \tilde{\bm{m}}_1\| \leq \|\bm{x} - {\bm{m}}_1\| + \|\bm{m}_1-\tilde{\bm{m}}_1\| \leq d_1 + (d^* - d_1) = d^*.\]
\item If $Q(\bm{x},\tilde{\mathcal{D}}) = \bm{m}_i$ for some $i = 2, \cdots, n$, and $Q(\bm{x},{\mathcal{D}}) = \bm{m}_j$ for some $j = 1, \cdots, n$, then it must be the case that $i = j$, since the only change from $\mathcal{D}$ to $\tilde{\mathcal{D}}$ is the first design point. Hence:
\[\tilde{d}^*  = \|\bm{x} -  {\bm{m}_i}\| \leq d_i \leq  d^*.\]
\ei 
This proves the proposition.

\pagebreak
\section{Minimax designs on $[0,1]^p$}

\begin{figure}[h]
\centering
\includegraphics[scale=0.50]{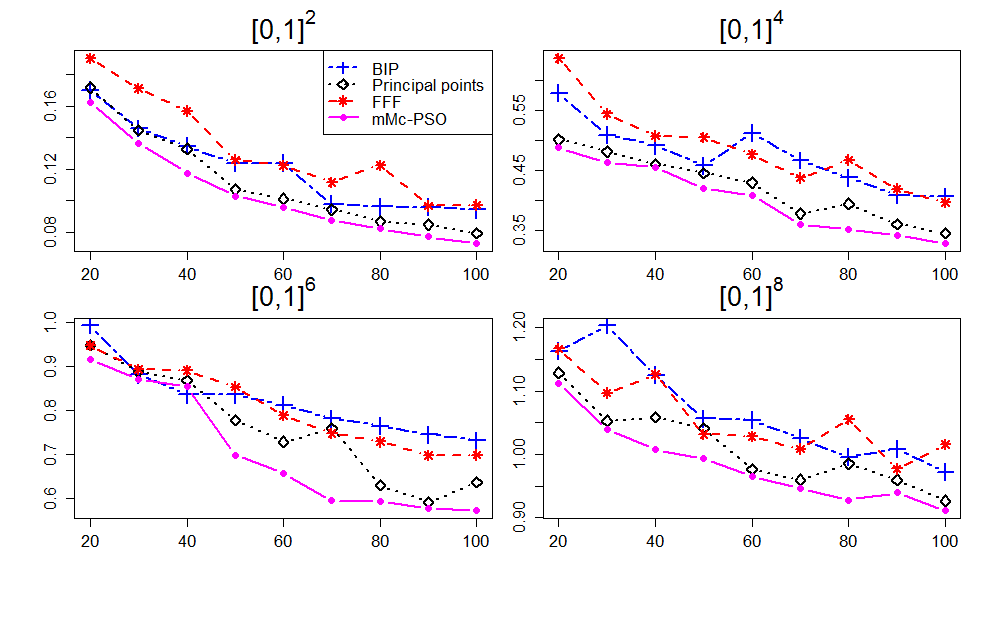}
\caption{Minimax criterion on $[0,1]^p$ for $p = 2, 4, 6$ and $8$.}
\end{figure}

\pagebreak

\begin{figure}[!htb]
\includegraphics[scale=0.50]{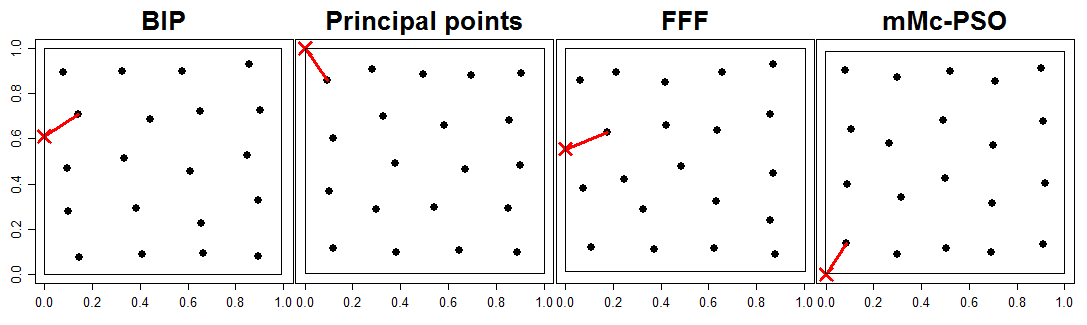}
\vspace{-0.2cm}
\includegraphics[scale=0.50]{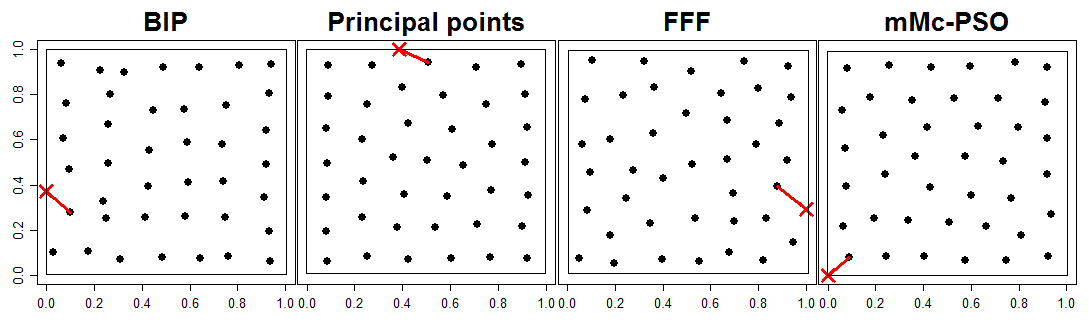}
\vspace{-0.2cm}
\includegraphics[scale=0.50]{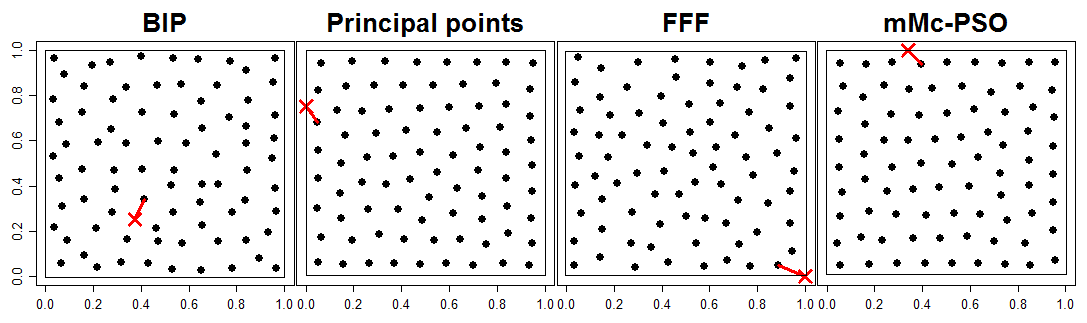}
\vspace{-0.2cm}
\includegraphics[scale=0.50]{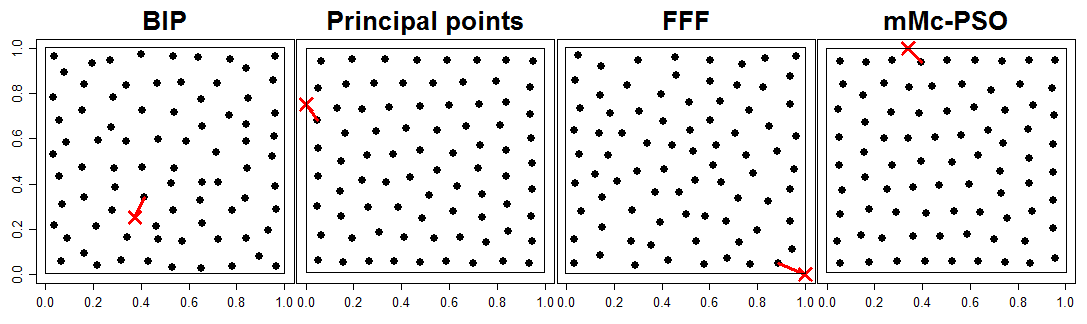}
\vspace{-0.2cm}
\includegraphics[scale=0.50]{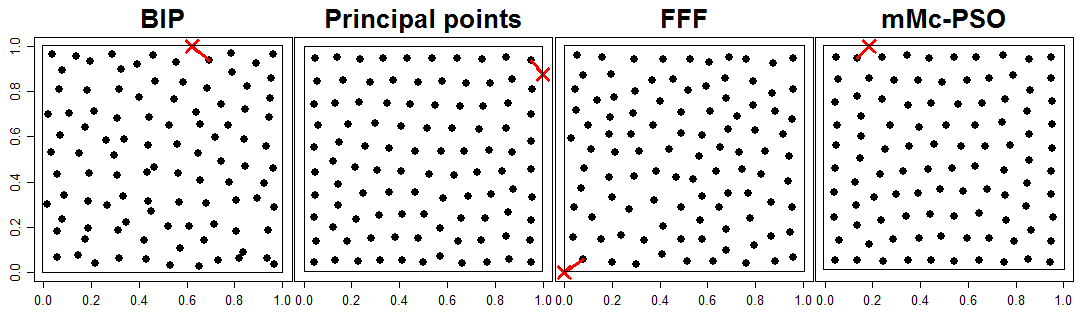}
\caption{20-, 40-, 60-, 80- and 100-point designs on the unit hypercube $[0,1]^2$.}
\end{figure}

\pagebreak

\section{Minimax designs on $A_p$ and $B_p$}

\begin{figure}[!htb]
\centering
\includegraphics[scale=0.50]{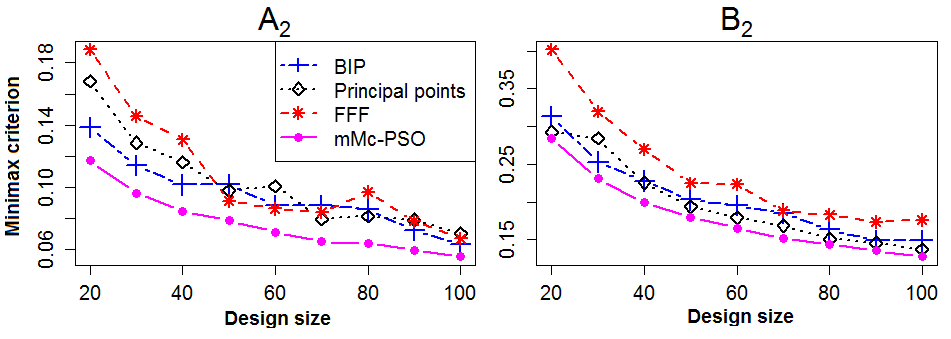}
\includegraphics[scale=0.45]{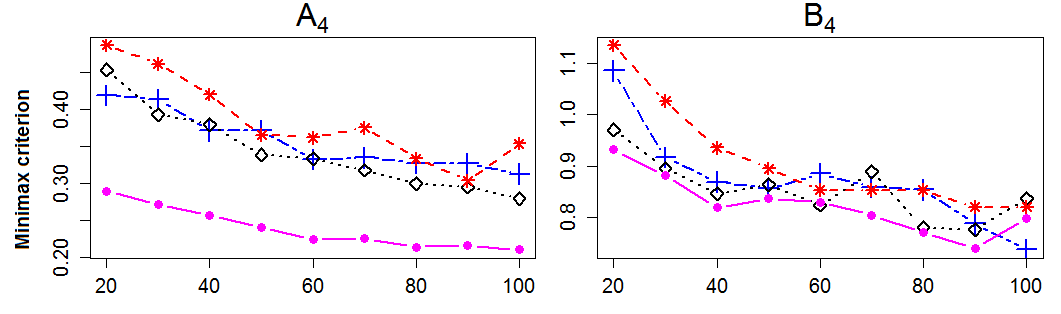}
\includegraphics[scale=0.45]{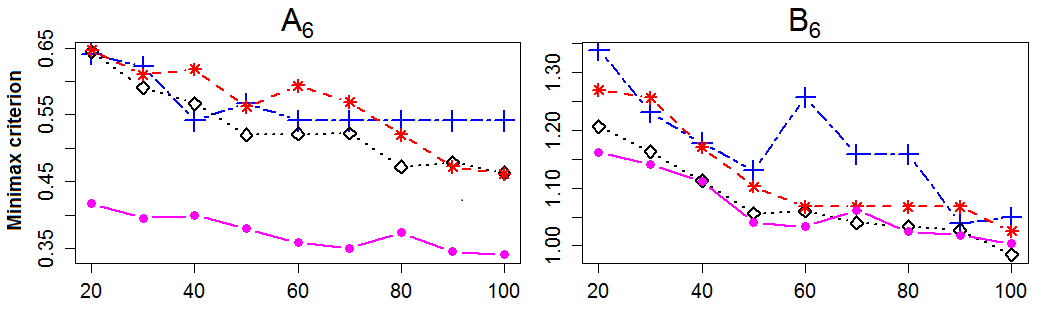}
\includegraphics[scale=0.50]{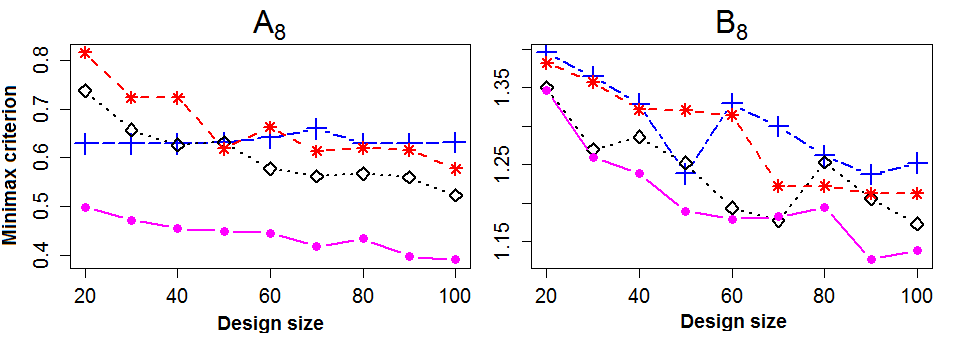}
\caption{Minimax criterion on $A_p$ and $B_p$ for $p = 2, 4, 6$ and $8$.}
\end{figure}

\pagebreak

\begin{figure}[!htb]
\includegraphics[scale=0.50]{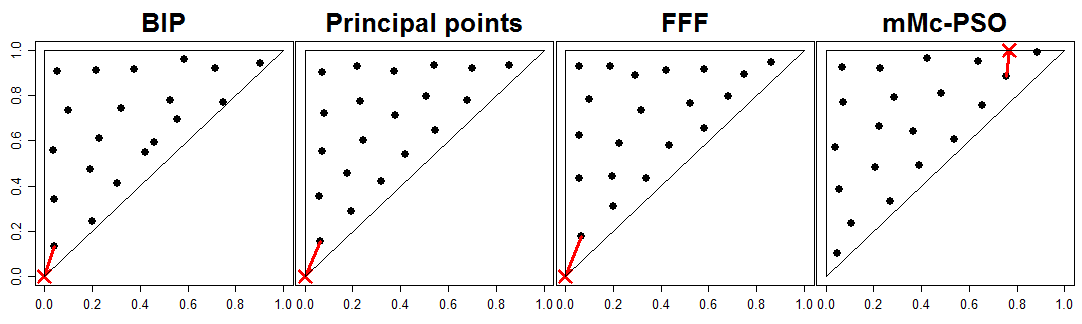}
\vspace{-0.2cm}
\includegraphics[scale=0.50]{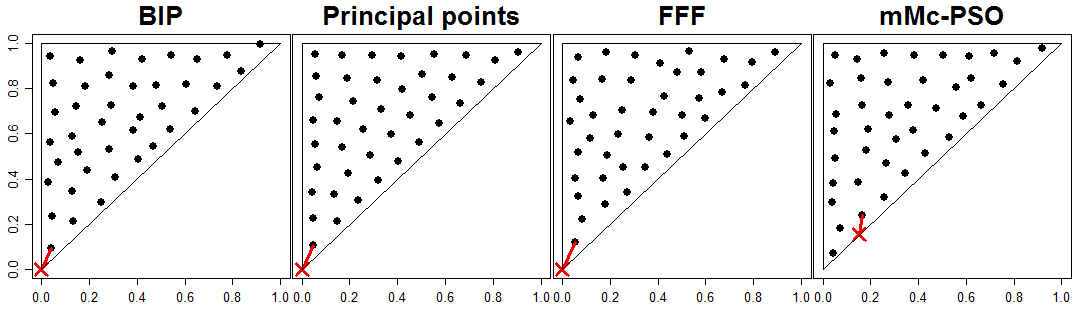}
\vspace{-0.2cm}
\includegraphics[scale=0.50]{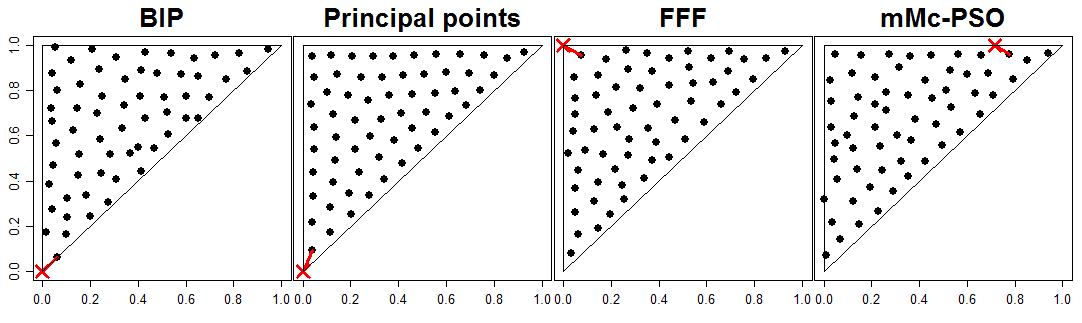}
\vspace{-0.2cm}
\includegraphics[scale=0.50]{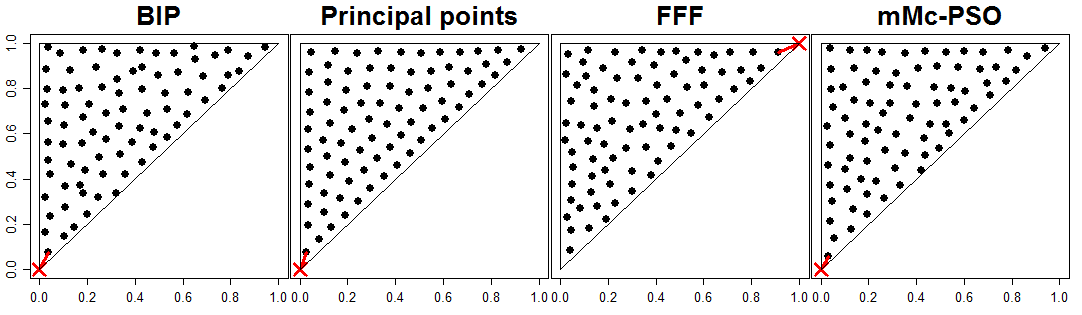}
\vspace{-0.2cm}
\includegraphics[scale=0.50]{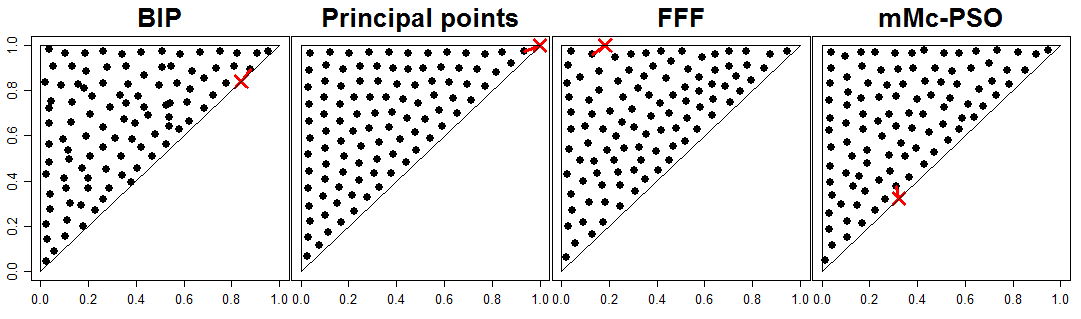}
\caption{20-, 40-, 60-, 80- and 100-point designs on the unit simplex $A_2$.}
\end{figure}

\pagebreak

\begin{figure}[!htb]
\includegraphics[scale=0.50]{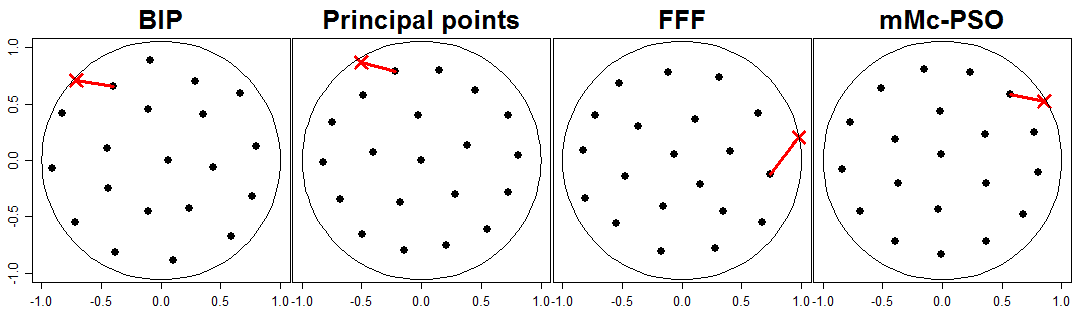}
\vspace{-0.2cm}
\includegraphics[scale=0.50]{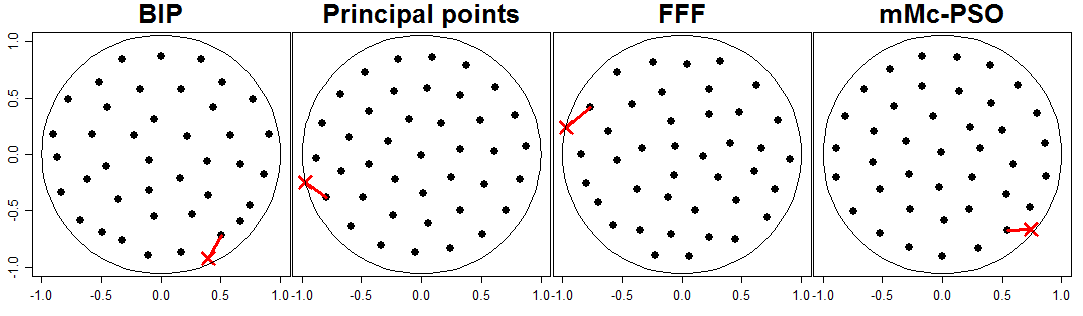}
\vspace{-0.2cm}
\includegraphics[scale=0.50]{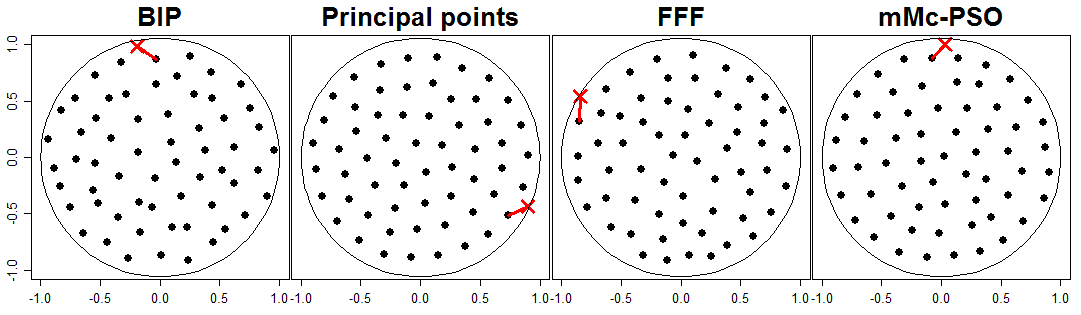}
\vspace{-0.2cm}
\includegraphics[scale=0.50]{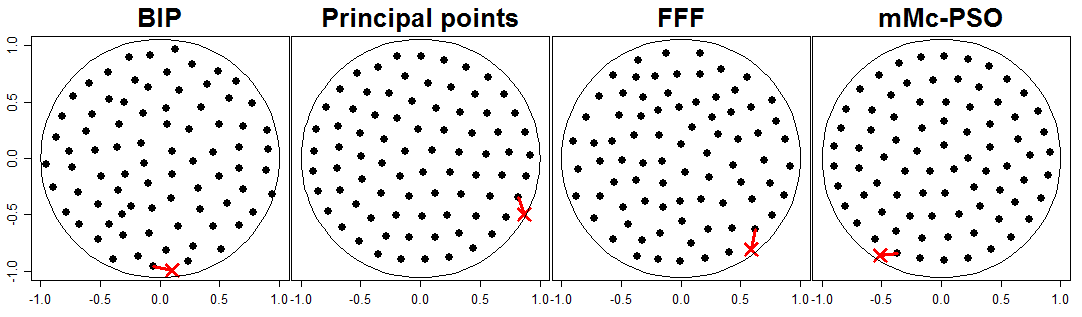}
\vspace{-0.2cm}
\includegraphics[scale=0.50]{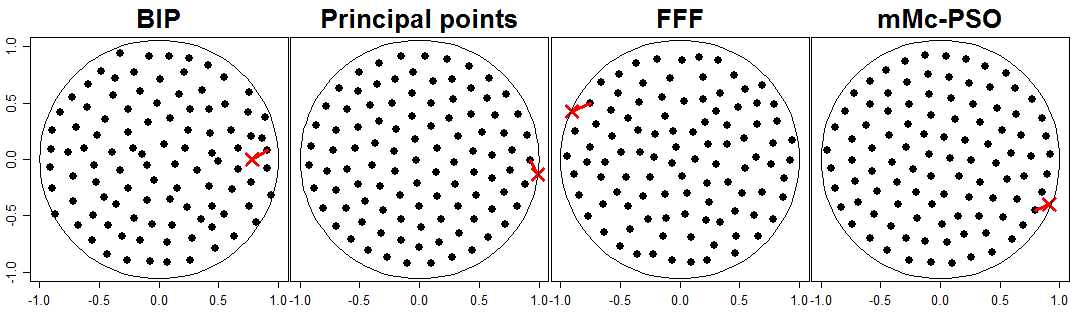}
\caption{20-, 40-, 60-, 80- and 100-point designs on the unit ball $B_2$.}
\end{figure}

\pagebreak
\section{Minimax designs on Georgia}
\begin{figure}[!htb]
\includegraphics[scale=0.45]{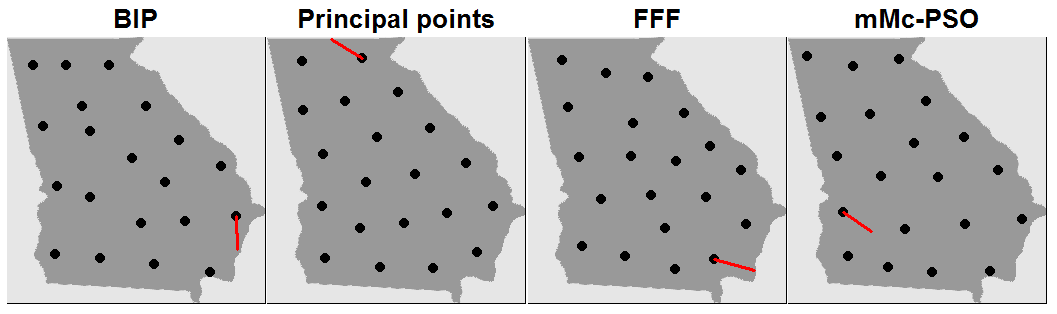}
\vspace{-0.2cm}
\includegraphics[scale=0.45]{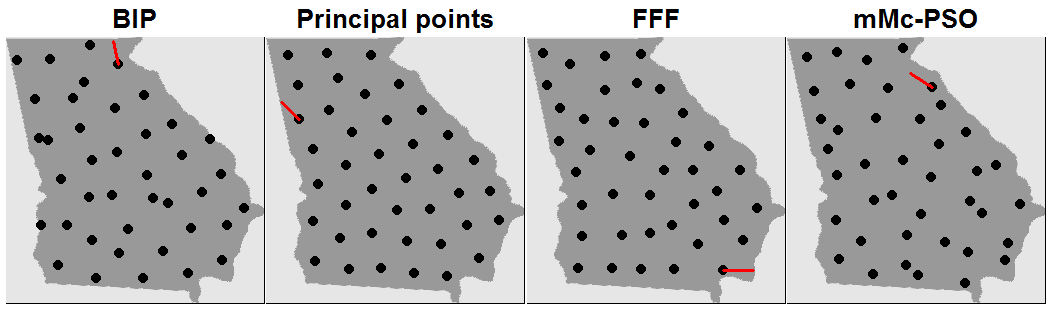}
\vspace{-0.2cm}
\includegraphics[scale=0.45]{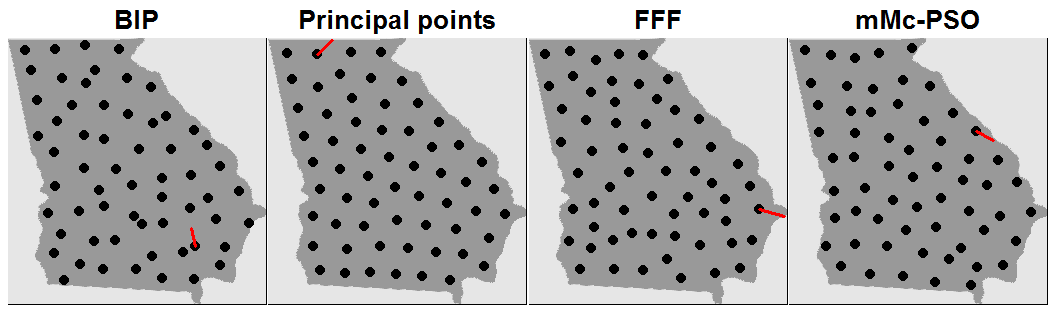}
\vspace{-0.2cm}
\includegraphics[scale=0.45]{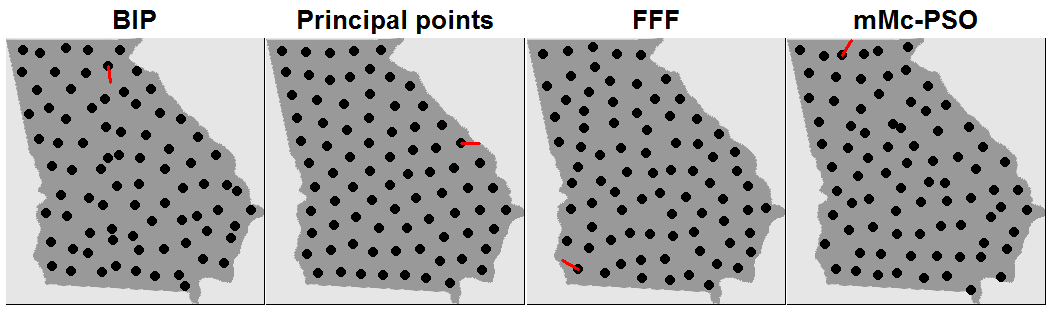}
\vspace{-0.2cm}
\includegraphics[scale=0.45]{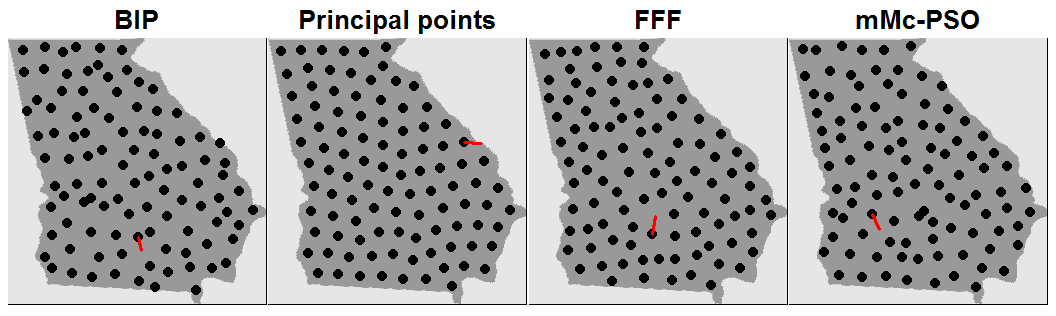}
\caption{20-, 40-, 60-, 80- and 100-point designs on Georgia.}
\end{figure}

\end{appendices}

\end{document}